\documentclass[reqno]{amsart}
\usepackage{mathtools}
\usepackage{microtype}
\usepackage[mathscr]{eucal}
\usepackage{slashed}
\usepackage{enumerate}
\usepackage{tikz}
\usepackage{amssymb}
\usepackage[colorlinks=true,linkcolor=black,citecolor=black]{hyperref}
\usepackage[sort&compress,capitalise,nameinlink]{cleveref}
\usepackage{ulem}
\usepackage{xcolor}

\usepackage{geometry}
\renewcommand{\geq}{\geqslant}
\renewcommand{\leq}{\leqslant}
\newcommand{\OO}{\mathcal{O}}

\newcommand{\R}{\mathbb R}
\newcommand{\N}{\mathbb N}
\renewcommand{\S}{\mathcal{S}}

\newcommand{\cB}{\mathcal B}

\newcommand{\Fb}{\mathbf F}
\newcommand{\tb}{\mathbf t}
\newcommand{\nb}{\mathbf n}
\newcommand{\Lap}{\mathscr L}
\newcommand{\dd}{\mathop{}\!{\mathrm d}}

\newcommand{\ii}{\mathrm i}
\newcommand{\ee}{\mathrm e}

\def\XXint#1#2#3{{\setbox0=\hbox{$#1{#2#3}{\int}$ }
\vcenter{\hbox{$#2#3$ }}\kern-.6\wd0}}

\newtheorem{theorem}{Theorem}[section]

\newtheorem{lemma}[theorem]{Lemma}
\newtheorem{proposition}[theorem]{Proposition}

\newtheorem{corollary}[theorem]{Corollary}
\theoremstyle{remark}
\newtheorem{remark}[theorem]{Remark}

\makeatletter
\@addtoreset{equation}{section}

\makeatother

\newcommand{\corner}{\Gamma_{\beta}(L,\ell)}

\newcommand{\wed}{W_{\beta}(L)}

\newcommand{\E}{\mathcal{E}}

\newcommand{\ecorr}{E_{\mathrm{corr}}}

\newcommand{\ecor}{E_{\mathrm{corner}}}
\newcommand{\ecorl}{E^{\mathrm{corn}}_{L,\ell}}

\newcommand{\domcor}{\mathscr{D}^{\mathrm{corn}}_{L,\ell}}

\newcommand{\ecordl}{E^{\mathrm{D}}_{L,\ell}}
\newcommand{\ecord}{E^{\mathrm{D}}_{\ell}}
\newcommand{\edir}{E^{\mathrm{D}}_{\mathrm{corner}}}
\newcommand{\domcord}{\mathscr{D}^{\mathrm{D}}_{L,\ell}}

\newcommand{\esec}{E_{\mathrm{sector}}}

\newcommand{\fsec}{\mathcal{E}_{\mathrm{sector}}}

\newcommand{\mcor}{\psi^{\mathrm{corn}}_{\mu}}
\newcommand{\msec}{\psi^{\mathrm{sec}}_{\mu}}
\newcommand{\mdirl}{\psi^{\mathrm{D}}_{L,\ell}}

\newcommand{\ewedl}{E^{\mathrm{wedge}}_L}
\newcommand{\domwed}{\mathscr{D}^{\mathrm{wedge}}_{L}}
\newcommand{\mwed}{\psi^{\mathrm{wedge}}_{\mu}}

\newcommand{\Hc}{H_{\mathrm{c}1}}
\newcommand{\Hcc}{H_{\mathrm{c}2}}
\newcommand{\Hccc}{H_{\mathrm{c}3}}
\newcommand{\Hstar}{H_{\mathrm{corner}}}

\numberwithin{equation}{section}
\newcommand{\bdm}{\begin{displaymath}}
\newcommand{\edm}{\end{displaymath}}
\newcommand{\bay}{\begin{array}{c}}
\newcommand{\eay}{\end{array}}
\newcommand{\ben}{\begin{enumerate}}
\newcommand{\een}{\end{enumerate}}
\newcommand{\beq}{\begin{equation}}
\newcommand{\eeq}{\end{equation}}
\newcommand{\beqn}{\begin{eqnarray}}
\newcommand{\eeqn}{\end{eqnarray}}
\newcommand{\bml}[1]{\begin{multline} #1 \end{multline}}
\newcommand{\bmln}[1]{\begin{multline*} #1 \end{multline*}}

\newcommand{\lf}{\left}
\newcommand{\ri}{\right}

\newcommand{\xv}{\mathbf{x}}
\newcommand{\pv}{\mathbf{p}}

\newcommand{\nv}{\mathbf{n}}
\newcommand{\fv}{\mathbf{F}}
\newcommand{\eps}{\varepsilon}
\newcommand{\diff}{\mathrm{d}}

\newcommand{\disp}{\displaystyle}
\newcommand{\tx}{\textstyle}

\newcommand{\gle}{E^{\mathrm{GL}}}
\newcommand{\glfe}{\mathcal{E}_{\eps}^{\mathrm{GL}}}
\newcommand{\glm}{\psi^{\mathrm{GL}}}

\newcommand{\aav}{\mathbf{A}}
\newcommand{\theo}{\Theta_0}
\newcommand{\aavm}{\mathbf{A}^{\mathrm{GL}}}
\newcommand{\fonea}{\mathcal{E}^{\mathrm{1D}}_{\alpha}}
\newcommand{\foneal}{\mathcal{E}^{\mathrm{1D}}_{\alpha, \ell}}
\newcommand{\eonea}{{E}^{\mathrm{1D}}_{\alpha}}
\newcommand{\eoneal}{{E}^{\mathrm{1D}}_{\alpha, \ell}}
\newcommand{\domone}{\mathscr{D}^{\mathrm{1D}}}
\newcommand{\eones}{E^{\mathrm{1D}}_{\star}}
\newcommand{\eonel}{E^{\mathrm{1D}}_{\ell}}
\newcommand{\glee}{E^{\mathrm{GL}}_{\eps}}

\title[{Ginzburg-Landau Energy of Corners}]{On the Ginzburg-Landau Energy of Corners}

\author[M. Correggi]{Michele Correggi}
\address{Dipartimento di Matematica,  Politecnico di Milano,   P.zza Leonardo da Vinci,  32,  20133, Milan,  Italy.}
\email{michele.correggi@gmail.com}

\author[E.L. Giacomelli]{Emanuela L.  Giacomelli}
\address{Department of Mathematics, LMU Münich, Theresienstr.  39,  80333,  Munich,  Germany.}
\email{emanuela.giacomelli@math.lmu.de}

\author[A.  Kachmar]{Ayman Kachmar}
\address{The Chinese University of Hong Kong (Shenzhen),  Shenzhen,  China.}
\email{akachmar@cuhk.edu.cn}

\begin{document}
\begin{abstract}
It is a well known fact that the geometry of a superconducting sample influences the distribution of the surface superconductivity for strong applied magnetic fields. For instance, the presence of corners  induces geometric terms described through effective models in sector-like regions. We study the connection between two effective models for the offset of superconductivity and for surface superconductivity introduced in \cite{BNF} and \cite{CG2}, respectively. We prove that the transition between the two models is continuous with respect to the magnetic field strength, and, as a byproduct, we deduce the existence of a minimizer at the threshold for both effective problems. Furthermore, as a consequence, we disprove a conjecture stated in \cite{CG2} concerning the dependence of the corner energy on the angle close to the threshold.
\end{abstract}
\maketitle

\section{Introduction and Main Results}

The phenomenon of surface superconductivity as well as the total loss of superconductivity in type-II materials due to the presence of intense applied magnetic fields is very well understood in the framework of the Ginzburg-Landau (GL) theory.  More in general, such an effective model is capable of accurately modelling the response of a type-II superconductor to the application of an external magnetic field. Let us briefly recall the different phases that can be observed in an elongated wire with a constant cross section $ \Omega \subset \R^2 $ and an applied magnetic field orthogonal to it: 
\begin{itemize}
	\item for small enough  magnetic fields no effect is observed and in fact the magnetic field lines are expelled from the superconductor (Meissner effect), until a first threshold $ \Hc $ is crossed and the field penetrates the sample at isolated defects (vortices);
	\item	for magnetic fields above a second critical field $ \Hcc $,  superconductivity is destroyed in the bulk of the sample but survives at the boundary (surface superconductivity); 
	\item 	eventually, for even stronger fields above a third threshold $ \Hccc $, superconductivity is lost everywhere and the material behaves as a conventional conductor again.
\end{itemize}

Such a behavior is even richer when corner singularities are present at the boundary of the sample. In this case, indeed, the transition to the normal state for increasing magnetic fields occurs in two steps:
\begin{itemize}
	\item above $ \Hcc $, the superconducting behavior is at first confined at the sample boundary and its distribution remains uniform there;
	\item then, above a precise threshold $ \Hstar $ in between $ \Hcc $ and $ \Hccc $, superconductivity concentrates at the corner of smallest opening angle and disappears everywhere else.
\end{itemize}

Let us now describe in more details the outlined phenomena and their description within GL theory: the equilibrium state of the sample is characterized by an order parameter $ \psi $ and an induced magnetic potential $ \aav $ minimizing the GL free energy, which above $ \Hcc $ in suitable units reads
\begin{equation}
	\label{eq: glf}
	\glfe [\psi, \aav]= \displaystyle\int_{\Omega} \diff\textbf{r}\; \bigg\{ \bigg| \left ( \nabla + i \frac{\aav}{\varepsilon ^ 2}\right)\psi \bigg|^2 -\frac{\mu}{2\varepsilon ^2}(2|\psi|^2-|\psi|^4)	 \bigg\}+\frac{1}{\varepsilon ^4}\displaystyle\int_{\R^2} \diff\mathbf{r}\; |\mbox{curl}\aav - 1|^2,
\end{equation}
where $ 0 < \eps \ll 1 $ is a small parameter depending on the material itself and $ \mu \in (0,1) $ is proportional to the inverse of the intensity of the applied field\footnote{In the literature it is often used $ b : = 1/\mu $ as a parameter to measure the applied field.}. Note that, as the intensity of the applied field increases, $ \mu $ decreases and therefore the critical fields are encountered for $ \mu $ decreasing from $ 1 $. In such units, the critical thresholds corresponding to  $ \Hcc = \mu_2 \eps^{-2} $ and $ \Hstar = \mu_{\mathrm{corner}}\eps^{-2} $ are
\beq
	\mu_2 = 1, 	\qquad		 \mu_{\mathrm{corner}} = \theo,
\eeq
respectively.  The threshold $\Hccc = \mu_3\eps^{-2}$, which is responsible for the total loss of the superconducting properties, is then shifted compared to domains with smooth boundaries: for regular boundaries indeed  $ \mu_3 = \mu_3(\pi) = \theo $; in the presence of corners $\mu_3= \mu_3(\beta) < \theo $ instead. Here $ \mu_3(\beta)$ is a suitable energy  depending on the opening angle $ \beta \in (0,\pi) $ of the smallest corner.   The constant $ \theo \in (0,1) $ is a universal value given by
\beq
	\label{eq: theo}
\theo = \inf_{\alpha\in\R} \inf_{\lf\| f \ri\|_{L^2(\R^+)}=1} \int_{0}^{+\infty} \diff t \lf\{ |\partial_t f|^2 + (t+\alpha)^2 |f|^2 \ri\}.
\eeq
We point out here that, although the transition from surface to corner superconductivity is believed to occur for any acute angle $ \beta $, the existence of $ \mu_{\mathrm{corner}}$ and the associated shift of $ \mu_3 $ is proven only for $ \beta $ in certain intervals. We will discuss this point further below.

The main quantity under investigation is therefore the infimum of the energy in \eqref{eq: glf} and the features of any associated minimizing pair $ (\glm, \aavm) $, if any, and, in particular, the region where $ \glm $ is concentrated. As anticipated, for $ \mu \in (\mu_{\mathrm{corner}},1) $, superconductivity is expected to be almost uniformly distributed along the boundary $ \partial \Omega $ and the effects of corners appear only as a higher order correction. This is made apparent by the energy asymptotics proven in \cite[Thm. 2.1]{CG2}. More specifically, considering for simplicity a sample with a single\footnote{Note however that the result in \cite{CG2} applies to the case of finitely-many corners.} corner at the boundary with opening angle $ \beta \in (0,2\pi) $, one has
\beq
	\label{eq: gle asympt corner}
	\glee = \disp\frac{|\partial\Omega| \eones}{\eps} -  \ecorr \int_{0}^{|\partial \Omega|} \diff \sigma  \: \mathfrak{K}(\sigma) + E_{\mathrm{corner}, \beta} + o(1), 	\qquad		\mbox{for } \mu \in \lf(\theo, 1 \ri).
\eeq
In the above asymptotics, $ \eones $ is obtained by minimizing a one-dimensional effective energy describing the variation of superconductivity in the normal direction to the boundary, $ \ecorr $ is a coefficient of order 1 expressed in terms of the one-dimensional minimizing profile, $\sigma $ is a tangential coordinate along the boundary and $ \mathfrak{K}(\sigma) $ is the boundary curvature. However, the most important term in the above expansion is the corner effective energy $ E_{\mathrm{corner}, \beta} $ (next-to-last term in \eqref{eq: gle asympt corner}), which yields the correction due to the presence of the corner and is defined by an implicit minimization problem (see \eqref{eq-E-CG} below).

A similar energy expansion holds in the corner superconductivity regime, i.e., when the superconducting behavior survives only close to the corner\footnote{More generally, one can consider a sample with finitely many corners at the boundary. In this case, the value $\beta$ corresponds to the smallest opening angle among all the corners at the boundary \cite{BNF,HK}.} at the boundary with opening angle $\beta$. In order to quantify such a behavior, we have first to discuss more in detail the threshold at which the transition takes place: in the same setting discussed above, it was proven in \cite[Thms. 1.4 \& 1.6]{BNF} that the critical value is explicitly given by
\beq
	\label{eq: Hcorner}
	\mu_3 = \mu_{3}(\beta) =: \mu_\beta,
\eeq
where
\beq
	\label{eq: mu beta}
	\mu_\beta = \inf_{\lf\| \psi \ri\|_2 = 1} \int_{\mathcal{S}_\beta} \diff \xv \: \lf| \lf(\nabla + \ii \fv \ri) \psi \ri|^2,
\eeq
is the bottom of the spectrum of a Schr\"{o}dinger operator with a uniform magnetic field of unit intensity in an infinite sector $ \mathcal{S}_{\beta} $ of opening angle $ \beta $. We set here 
\beq
	\fv(\xv) : = \tfrac{1}{2} \xv^{\perp},		\qquad		\xv^{\perp} : = (-y,x),
\eeq
for $ \xv = (x,y) $ and
\beq
	\label{eq: sector}
	\mathcal{S}_\beta = \lf\{ \xv \in \R^2 \: | \: \pi - \beta < \vartheta(\xv) < \pi \ri\},
\eeq
where we used polar coordinates $ \varrho = |\xv| \in \R^+ $ and $ \vartheta = \arg (x+\ii y) \in [0, 2\pi) $.  In \eqref{eq: mu beta}, we make no assumption on $ \psi $ at the boundary of $ \mathcal{S}_\beta $, i.e., we consider the Neumann realization of the magnetic Schr\"{o}dinger operator. Such a result is however conditioned to the assumption (see \cite[Ass. 1.3]{BNF})
\beq
	\label{eq: condition}
	\mu_\beta  < \theo,
\eeq
whose general proof is still lacking, although numerical simulations strongly suggest that it should hold for any $\beta \in (0, \pi) $ (see  \cite{ABN} and \cite[Chpt. 8]{Ra}). Note that the above conjecture can be reformulated in short by stating that the lowest eigenvalue of a magnetic Schr\"{o}dinger operator in any sector with acute opening angle is strictly below the lowest eigenvalues in the half-plane, which equals $ \theo $ (see, e.g., \cite[Thm. 8.1.1]{FH1}). There are however some intervals of $ \beta $ in which \eqref{eq: condition} is shown to hold:
\begin{itemize}
	\item for $ 0 < \beta < \pi/2 + \epsilon_1 $, for some given  $ \epsilon_1 > 0 $ \cite{Bo,Ja,ELP};
	\item for $ \pi - \epsilon_2 < \beta < \pi $, for some given  $ \epsilon_2 > 0 $ \cite{BNFKR}.
\end{itemize}
Note that in presence of more than one corner with acute angles, one observes \cite{HK} the occurrence of several critical values $ \mu_{\mathrm{corner}, 1} > \mu_{\mathrm{corner}, 2} > \ldots $  marking the transition of superconductivity concentration from one corner to another {according to the opening angle (the smaller the angle is, the more superconductivity is attracted)}.

Concerning the GL energy in the corner regime, the main result is stated in \cite[Thm. 1.7]{BNF}: assuming that $ \partial \Omega $ contains only one corner with {angle $ \beta $  such that $\mu_\beta < \Theta_0$}, one has\footnote{{As mentioned  above, the theorem extends to the case of finitely many corners with opening angles $ \beta_j $ such that the corresponding spectral value $\mu_{\beta_j}$ is  below $\Theta_0$.}}
\beq
	\label{eq: gle asympt sector}
	\glee = E_{\mathrm{sector}, \beta} + o(1), \qquad		\mbox{for } \mu \in  \lf(\mu_{3, \beta}, \theo\ri),
\eeq
where $ E_{\mathrm{sector}, \beta} $ is obtained via the minimization of a suitable nonlinear functional in an infinite sector with angle $ \beta $. This is to some extent surprising, given that the closer $ \mu $ is to $ \mu_{3, \beta} $ the more relevant is the linear part of the GL functional, so that the threshold itself $ \mu_{3, \beta} $ is completely determined by the linear problem.

{In this paper we focus on the two effective models leading to the corner energies $ E_{\mathrm{corner}, \beta} $ and $ E_{\mathrm{sector}, \beta} $ (see their precise definitions below), which share some common features, although they characterise the superconducting behaviour in different regions and regimes.} By investigating the relation between the two effective models, we also expect to shed some light on both of them. As we are going to see, the variational problem leading to $ E_{\mathrm{corner}, \beta} $ is rather implicit and it is not even clear whether $ E_{\mathrm{corner}, \beta} \neq 0 $ for generic value of $ \beta $. Furthermore, a conjecture on the explicit dependence of $ E_{\mathrm{corner}, \beta} $ on $ \beta $ is stated in \cite[Conj. 1]{CG2} and confirmed to leading order for almost flat angles in \cite{CG3}. Therefore, any connection between the two variational problems may easily provide nontrivial information on open related questions.

Before giving a precise definition of the two effective models, let us specify the common setting: we fix the corner angle $ \beta \in (0, \pi) $ once for all, so that we can drop the dependence on $ \beta $ in the notation, and assume that \eqref{eq: condition} holds. Both effective models are then well-posed and yield two energies $ \ecor(\mu) $ and  $ \esec(\mu) $ depending on the applied magnetic field through $ \mu \in (\mu_3,1) $. Of course, the first one is defined for $ \mu \in (\theo,1) $, while the second problem is set for $ \mu \in (\mu_3,\theo) $. {A comparison is therefore only possible at the threshold $\theo $ and, to anticipate our main result, we will indeed prove that both quantities coincide there.}

\subsection{Effective models}

Let us give the precise definition of the two effective functionals. We start from the one leading to $ E_{\mathrm{sector}}(\mu) $, which is expected to capture the superconducting behavior in the corner region.
 For any
\beq
	\mu_{3} < \mu < \theo,
\eeq
{we define:}
\begin{eqnarray}
\label{eq:ref-en}
 	\fsec[\psi]  &=& \E_\mu[\psi;\S_{\beta}],		\\
 	\E_\mu[\psi;\S] & = & \int_{\mathcal{S}} \diff \xv \: \left[ \lf| \lf(\nabla + \ii \fv \ri) \psi \ri|^2 - \mu|\psi|^2 + \tfrac{1}{2} \mu |\psi|^4\right],	
\label{eq:ref-en 2}
\end{eqnarray}
{where  $ \S_{\beta} $ is the infinite sector introduce in \eqref{eq: sector} and depicted in \cref{fig: sector}. The energy functional $\fsec$ is defined for $\psi $ in the magnetic Sobolev space 
\begin{equation}\label{eq:sob-mag}
W^{1,2}_{\Fb}(\mathcal{S}_{\beta}) := \lf\{ \psi\in L^2(\mathcal{S}_{\beta}) \: \big| \: (\nabla-\ii\Fb)\psi \in L^2(\mathcal{S}_{\beta}) \ri\}.
\end{equation}
Note that in usual samples the GL order parameter is proven to decay in the distance from the corner on the scale $ \eps $ for any $\mu_{3} < \mu < \theo$ \cite[Thm. 1.6]{BNF} and therefore the key contribution to the energy comes from a blow-up of the {corner} region, which leads to the effective energy defined in \eqref{eq:ref-en}.}

The sector energy is given by the ground state energy of $ \fsec $, i.e.,
\begin{equation}\label{eq:ref-en*}
 	\esec(\mu) = \inf_{\psi \in W^{1,2}_{\Fb}(\mathcal{S}_{\beta})} \fsec[\psi],
\end{equation}
and we denote by $ \msec \in W^{1,2}_{\Fb}(\mathcal{S}_{\beta}) $ any corresponding minimizer.

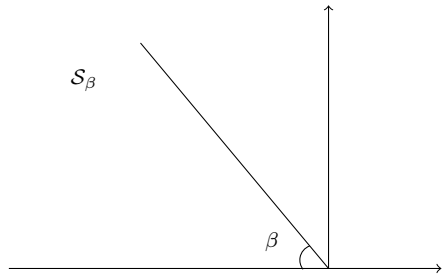
\begin{figure}[!ht]
		\begin{center}
		\begin{tikzpicture}[scale=0.5]
			\draw[->] (2,2) to (13.5,2);
			\draw[->](10.5,2) to (10.5,9);
			\draw (5.5,8) to (10.5,2);
			\draw (10,2.6) arc (110:218:0.4);
			\node at (4,7) {{\footnotesize $\mathcal{S}_\beta$}};
			\node at (9,2.7) {{\footnotesize $\beta$}};
		\end{tikzpicture}
		\caption{The infinite sector $ \mathcal{S}_\beta $ with opening angle $ \beta $ and the corresponding coordinate frame.}\label{fig: sector}
		\end{center}{}
	\end{figure}

In order to introduce the corner effective model, we have first to discuss {a one-dimensional model problem on the half-line (see e.g. \cite{CR1,CG1}) and, later, a corresponding one on a finite segment} (see \cite[Appendix A]{CG2} for further details). Such effective models are known to capture the superconducting features of the material close to boundary in the surface superconductivity regime, where the one-dimensional coordinate is the normal to the boundary itself. More precisely, in this regime the GL order parameter is such that $ |\glm(\xv)| \simeq f(\mathrm{dist}(\xv, \partial \Omega)) $, i.e., its modulus is to leading order constant along the tangential direction and decays in the distance from the boundary. The profile $ f $ is an approximate minimizer of the one-dimensional energies mentioned above. {In the case of the half-line,} after blow-up on the scale $ \eps $ of the boundary region and the extraction of a suitable local gauge phase, the energy per unit tangential length is approximated by the one-dimensional functional
\beq
	\label{eq: fonea}
	\fonea[f] = \int_0^{+\infty} \diff t \: \lf\{ {f^{\prime}}^2 + (t+\alpha)^2 f^2 - \mu f^2 +\tfrac{1}{2} \mu f^4 \ri\},
\eeq
where $ f \in \domone : = \{ f: \R^+ \to \R ~:~ H^1(\R^+) \cap L^2(\R_+, t^2 \dd t) \}  $ and \(\alpha\in\R\) is a parameter determined by the tangential current along $ \partial \Omega $ of $ \glm $. After joint minimization w.r.t. both $ f $ and $ \alpha $, one gets the energy
\begin{equation}\label{eq:gse-1D}
	\eones(\mu) = \inf_{\alpha\in\R} \inf_{f \in \domone} \fonea[f]
\end{equation}
and we denote by $ f_{\star} $ and $ \alpha_{\star} $ any minimizing pair. Since we are going to consider a finite sector, we introduce also the analogue of \eqref{eq: fonea} on the interval $ [0,\ell] $, $ \ell > 0 $, i.e.,
\beq\label{eq: foneal}
	\foneal[f] = \int_0^{\ell} \diff t \: \lf\{ {f^{\prime}}^2 + (t+\alpha)^2 f^2 - \mu f^2 +\tfrac{1}{2} \mu f^4 \ri\},
\eeq
with ground state energy 
\begin{equation}\label{eq:gse-1D0}
\eonel(\mu) = \inf_{\alpha\in\R} \inf_{f\in H^1((0,\ell); \R)} \foneal[f].
\end{equation}
Any corresponding minimizing pair is denoted by $ (f_0, \alpha_0) $.

We are now able to provide the definition of $ \ecor $ which is the outcome of the minimization of a suitable energy in a domain $ \corner $ supposed to reproduce a rescaling of the boundary region close to the corner. {The domain $\corner$} is depicted in \cref{fig: corner}, its explicit definition requires the use of tubular coordinates. There is a natural parameterization {$\pv:\R\to \partial \S_\beta$ of the boundary of the sector $\S_\beta$ with center at the vertex, such that}
$\pv(s)= (-s,0)$, if $s\leq 0$, and such that $\pv(s)$ lies on the line {$\{(x,-x\tan\beta)\}$,} if $s>0$.  Let us introduce a direct frame $(\tb,\nb)$ at $\pv(s)$,  defined everywhere on $\partial \S_\beta$ except at the vertex such that $\tb$ is tangent to $\partial \S_\beta$ and $\nb$ is normal to $\partial \S_\beta$ and points inward $\S_\beta$; for $\pv(s)=(s,0)$ and $s<0$,  we have $\tb=(-1,0)$ and $\nb=(0,1)$.
Then,  we introduce a parametrization $ \xv: \R\setminus\{0\} \times\R_+ \to \S_\beta $ given by
\beq
	\xv(s,t) := s\tb + t\nb.
\eeq 
Note that
\beq\label{eq:def-dist-t}
	t(\xv) : = \mathrm{dist} \lf(\xv, \partial \mathcal{S}_\beta \ri),
\eeq
is the normal distance to the boundary of $ \S_{\beta} $, while {$s(\xv)$ is} the tangential distance  $|\mathbf{p}(s(\xv))-\xv|=\mathrm{dist} \lf(\xv, \partial \mathcal{S}_{\beta} \ri)$, provided that $\xv$ is not on the bisectrix of $\S_\beta$. 

For $L>0$ and $0<\ell<L\tan(\beta/2)$,  the domain $\corner$ is then defined as 
\begin{equation}\label{eq: corner}
\corner = \lf\{ \xv(s,t) \: \big| \: |s|<L,   0 < t< \ell \ri\}.
\end{equation} 
The boundary of $ \corner $ consists  of three parts
\[ \begin{aligned}
\partial \Gamma^{\rm out} &=\overline{\corner}\cap \partial \S_{\beta},\\
 \partial \Gamma^{\rm bd} &=\overline{\corner}\cap \{s(\xv)=\pm L\},	\\  		\partial \Gamma^{\rm in} &= \overline{\corner}\cap \{t(\xv)=\ell\}.\end{aligned}
 \]
Under the following conditions on the parameters $ L, \ell $,
\beq
	\label{eq: L ell condition}
	1\ll \ell\ll L\ll \ell^a, \qquad \mbox{for some }	a> 1,
\eeq
and assuming that
\beq
	\theo < \mu < 1,
\eeq
we introduce the following ground state energy
\begin{eqnarray}\label{eq:gse-corner}
 	\ecorl(\mu) &=& \inf_{\psi\in \domcor}\mathcal \E_\mu[\psi; \corner]\nonumber \\
 	&=& \inf_{\psi\in \domcor}\mathcal \int_{\corner} \diff \xv \: \left[ \lf| \lf(\nabla + \ii \fv \ri) \psi \ri|^2 - \mu|\psi|^2 + \tfrac{1}{2} \mu |\psi|^4\right].
\end{eqnarray}
We denote by $ \mcor $ any associated minimizer. Here,
\beq
	\domcor = \lf\{\psi\in H^1(\R^2) \: \big| \: \psi=\psi_{0} \mbox{ on }\partial\Gamma^{\rm bd} \cup \partial\Gamma^{\rm in}  \ri\}
\eeq
and \(\psi_{0} \) can be expressed in tubular coordinates in terms of the minimizing pair $ ( f_0, \alpha_0) $ as 
\beq
	\label{eq: psistar}
	\psi_{0}(\xv(s,t)) = f_0(t) e^{\ii\alpha_0 s -\frac12 \ii st }.
\eeq

Heuristically, the corner energy is obtained by removing from the ground state energy in \eqref{eq:gse-corner} the contribution to the almost uniform distribution of superconductivity along the (outer) boundary of the corner region, which is in turn provided by the one-dimensional energy $ \eones $ times the length of the outer boundary $ 2 L $. More precisely, $ \ecor $ is given by the following limit
\beq
	\label{eq-E-CG}
	\ecor(\mu) = \lim_{\substack{\ell\to+\infty\\ \ell\ll L\ll \ell^a}} \lf( -2L \eonel(\mu) + \ecorl(\mu) \ri)
\end{equation}
which is proven to exist and being independent of \(a>1\) in \cite[Prop. 2.2]{CG2} (see also \cite[Rem. 3.8]{CG2}).

\begin{figure}[!ht]
		\begin{center}
		\begin{tikzpicture}[scale=0.75]
			\draw (0,0) -- (1,2);
			\draw (1,2) -- (1.5,3);
			\draw (1.5,3) -- (2,4) -- (2.5,3);
			\draw (2.5,3) -- (3,2);
			\draw (3,2) -- (4,0);
			\draw (0,0) -- (1,-0.5);
			\draw (4,0) -- (3,-0.5);
			\draw (1,-0.5) -- (2,1.7);
			\draw (3,-0.5) -- (2,1.7);
			\node at (2,4.5) {{\footnotesize $V$}};
			\node at (-0.5, 0) {{\footnotesize $A$}};
			\node at (4.5, 0) {{\footnotesize $B$}};
			\node at (1.5, -0.5) {{\footnotesize $C$}};
			\node at (2.5, -0.5) {{\footnotesize $E$}};
			\node at (2,1) {{\footnotesize $D$}};
		\end{tikzpicture}
		\hspace{0,5cm}
		\begin{tikzpicture}[scale=0.45]
			\draw[->] (0.5,2) to (13.5,2);
			\draw[->](10.5,0.5) to (10.5,9);
			\draw (2.4,2) to (2.4,3.6);
			\draw (2.4,3.6) to (6.9,3.6);
			\draw (4,7) to (6.9,3.6);
			\draw (4,7) to (5.5, 8);
			\draw (5.5,8) to (10.5,2);
			\node at (11.4,1.5) {{\footnotesize $x$}};
			\node at (9.9,8.7) {{\footnotesize $y$}};
		\end{tikzpicture}
		\caption{The corner region $ \corner $ and the associated coordinate system. The opening angle $ \widehat{AVB} $ is equal to $ \beta $ and the side lengths are $ |\overline{AV}| = |\overline{VB}| = L  $ and $  |\overline{AC}| = |\overline{EB}| = \ell $.}\label{fig: corner}
		\end{center}{}
	\end{figure}
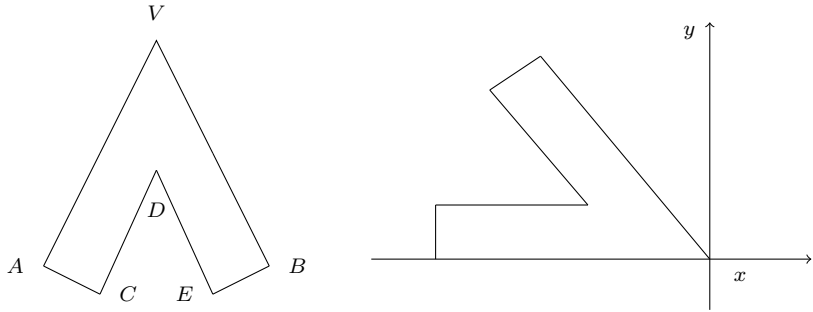
	
\subsection{Main result}

As anticipated we aim at making a comparison between the two effective energies, which of course can be done only at the threshold $ \mu = \theo $ for the transition from boundary to corner superconductivity. Both effective models are however not defined at $ \theo $ and therefore part of the problem is to extend their definitions exactly at the threshold via a limit procedure. Our main result is then stated below.

	\begin{theorem}[Effective energies]
	\label{theo: main}
	\mbox{}	\\
	For any angle $ \beta \in (0, \pi) $ such that \eqref{eq: condition} holds,  we have that
	\beq
		\label{eq: continuity}
		\lim_{\mu \to \theo^+} \ecor(\mu)= \lim_{\mu \to \theo^-} \esec(\mu)=\esec(\theo).
	\eeq
	\end{theorem}
The coincide of the two limits has some direct but important consequence on the two variational problems at the threshold, in particular concerning the existence of minimizers.
	
	\begin{corollary}
	\label{corollary}
	\mbox{}	\\
	Under the assumptions of \cref{theo: main}, $\esec(\theo)$ is finite and negative, and the variational problem \eqref{eq:ref-en*} at $ \mu =\theo $ admits a minimizer.
	\end{corollary}
	\begin{remark}[Non-triviality of $ \ecor $]
		\mbox{}	\\
		As stressed in \cite{CG2}, due to its implicit definition, it is not apparent that $ \ecor $ is non-trivial, i.e., a priori it might well be that $ \ecor = 0 $. For angles close to $ \pi $, it is proven that this is not the case in \cite{CG3}, but the question remains open otherwise. The continuity stated in \eqref{eq: continuity} guarantees however that $ \ecor $ is not zero {\it for any angle $ \beta $} and for $ \mu $ close enough to $ \theo $, since the r.h.s. of \eqref{eq: continuity} is  non-zero.
	\end{remark}
\begin{remark}[Continuity of $ \gle $]
		\mbox{}	\\
		Another important consequence of \eqref{eq: continuity} is that the GL energy $ \gle $ in a domain with a single corner {(whose opening angle satisfies \eqref{eq: condition})} at the boundary is a continuous function of $ \mu $ in a neighbourhood of $ \theo $ for fixed $ \eps $. Indeed, since (see \cite[Sect. 14.2.2]{FH1} and \cite[Proof of Lemma 2.3]{CR3})
		\beq\label{eq:(1.30)}
			\lim_{\mu \to \theo^+} \eones(\mu) = 0,	\qquad		\lim_{\mu \to \theo^+} \ecorr(\mu) = 0,
		\eeq
		for fixed $ \eps \in (0,1) $, the asymptotic expansion \eqref{eq: gle asympt corner} varies continuously as $ \mu $ crosses $ \theo $ from above becoming \eqref{eq: gle asympt sector} below the threshold.
	\end{remark}		
\begin{remark}[Conjecture on $ \ecor $]\mbox{ }\\
In \cite{CG2} it is conjectured that $ \ecor $ is in fact a linear function of the angle $ \beta $, namely that $\ecor(\mu)=-(\pi-\beta)E_{\rm corr}(\mu)$.  This conjecture was approximately confirmed for angles close to $ \pi $ in \cite{CG3}.  However,  our main result disproves the conjecture when $\mu$ is close to $\Theta_0$,  since $\esec(\Theta_0)<0$ according to \cref{corollary},    while  $E_{\rm corr}(\mu)$ vanishes for $\mu=\Theta_0$ and is positive for $\mu<\Theta_0$  (see  \eqref{eq:(1.30)}, \cite{CDR} and  \cite[Eq.~(2.17)]{Co}).    
\end{remark}	
	\begin{remark}[Magnetic steps]\label{rem:mag steps}\mbox{ }\\
	There has been an interest in studying the influence of discontinuous magnetic fields (magnetic steps) on the distribution of superconductivity \cite{A20,  AKP},  where the discontinuity  line of the magnetic field  causes similar effects as those due to corners \cite{A21,  AG, Gia,  M}.  We expect a  continuity result similar to \cref{theo: main} for the effective energies in the setting of discontinuous magnetic fields.
	\end{remark}	

	\subsection{Organization of the paper}	
The rest of the paper is divided into four sections.  {In \cref{sec:pre}, we collect some preliminary results, mostly related to effective one dimensional problems and the main features of the magnetic Laplacian in domains with corners, that will be used throughout the paper.  In \cref{sec:cor}, we study the corner effective energy and establish that  the limit $\lim_{\mu \to \theo^+} \ecor(\mu)$ exists and can be obtained by considering a modified variational problem in a solid wedge with zero boundary conditions. In \cref{sec:sec}, we study the sector effective energy and prove the existence of the limit $ \lim_{\mu \to \theo^-} \esec(\mu)$. Finally, in \cref{sec: proof of main thm}, we conclude the proof of \cref{theo: main}, comparing the two energies at the threshold.}

\section{Preliminaries}\label{sec:pre}

We collect here all the proofs of the technical results needed to get our main result \cref{theo: main}, whose proof is spelled in the next section. First, in Section \ref{subsec: 1d energy}, we recall some known facts and state some simple results concerning the one-dimensional effective problems $\fonea$ and $\foneal$ {introduced in \eqref{eq: fonea} and \eqref{eq: foneal}, respectively. After that,  we discuss some very well known properties of the magnetic Laplacian with uniform magnetic field in Section \ref{subsect: magnetic lapl}. Finally, in Section \ref{subsec: Dirichlet eff energy}, we introduce another variation problem in the corner region $\Gamma_\beta(L,\ell)$ which is going to be useful in the next section (see \cref{prop:mu=Th0}) to characterize the value of $E_{\mathrm{corner}}(\mu)$ at the threshold $\mu = \Theta_0$.}

Before starting our discussion, we observe that all the effective two-dimensional problems we are going to consider are related to the functional $ \E_{\mu}[\psi; \S] $ in a given domain $ \S $ as defined in \eqref{eq:ref-en 2}:
\[
	\E_\mu[\psi;\S] =\int_{\S} \diff \xv \: \left[ \lf| \lf(\nabla + \ii \fv \ri) \psi \ri|^2 - \mu|\psi|^2 + \tfrac{1}{2} \mu |\psi|^4\right].
\] 
It is very well known that,   irrespective of the type of boundary conditions on $ \S $ (provided that in the case of non-zero Dirichlet conditions the same bound holds on $ \partial \S $), one can prove in full generality that any minimizer $ \psi_{\sharp} $ of such a functional satisfies
\beq
	\label{eq: infty}
	\lf\| \psi_{\sharp} \ri\|_{L^{\infty}(\S)} \leq 1,
\eeq
by a direct application of the maximum principle. A similar bound holds also for the one-dimensional functional studied below.

\subsection{One-dimensional energies} \label{subsec: 1d energy}
{We recall that $ \eones(\mu) = \inf_{\alpha\in\R} \inf_{f \in \domone} \fonea[f]$, with 
\[
	\fonea[f] = \int_0^{+\infty} \diff t \: \lf\{ {f^{\prime}}^2 + (t+\alpha)^2 f^2 - \mu f^2 +\tfrac{1}{2} \mu f^4 \ri\},
\]
where $ f \in \domone : = \{ f: \R^+ \to \R \: | \: H^1(\R^+) \cap L^2(\R_+, t^2 \dd t) \}  $ and \(\alpha\in\R\).
For $ \mu \in (\theo, 1) $, it is known that $ \eones < 0 $ is realized on (at least one) minimizing pair $ (f_{\star}, \alpha_{\star}) $, which is non-trivial and such that the following optimality condition holds}
\beq
	\label{eq: optimality}
	\int_{0}^{+\infty} \diff t \: (t + \alpha_{\star}) f_{\star}^2 = 0.
\eeq
On the opposite, the minimization of $ \fonea $ becomes trivial for $ \mu \leq \theo $: the obvious upper bound $ \eones(\mu) \leq 0 $ can indeed be matched by the lower bound
\bdm
	\eones(\mu) \geq \inf_{f \in \domone} \lf[ \lf( \theo - \mu \ri) \lf\| f \ri\|_2^2 + \tfrac{1}{2} \mu \lf\| f  \ri\|_4^4 \ri] = 0,
\edm
for any $ \mu \leq \theo $, since by \eqref{eq: theo} the bottom of the spectrum of the shifted harmonic oscillator is bounded from below by $ \theo $ for any $ \alpha \in \R $. Hence,
\beq\label{eq: E* and f* at Theta0}
	\eones(\mu) = 0,		\qquad		f_{\star} \equiv 0, 	\qquad		\mbox{for any } \mu \leq \theo.
\eeq
The optimal phase is then undetermined and any finite $ \alpha $ yields the same energy. Moreover, when $ \mu \to \theo^+ $, it is known (see \cite[Lemma 14.2.11 and Prop. 14.2.12]{FH1} and \cite[Proof of Lemma 2.3]{CR3}) that
\beq
	\label{eq: asympt mu theo}
	\eones(\mu) = - \frac{\theo}{2} \frac{(\mu - \theo)^2}{\lf\| \psi_0 \ri\|_4^4} \lf(1 + o_{\mu - \theo}(1) \ri),		\qquad		\lf\| f_{\star} - \tfrac{\sqrt{\mu - \theo}}{\theo \lf\| \psi_0 \ri\|_4^2} \psi_0 \ri\|_{\infty} = o\lf(\sqrt{\mu - \theo}\ri),
\eeq
where $ \psi_0 $ stands for the normalized ground state of the shifted harmonic oscillator $ - \partial_t^2 + (t - \sqrt{\theo})^2 $ on $ L^2(\R^+) $, i.e., the state realizing \eqref{eq: theo}. Note that combining \eqref{eq: asympt mu theo} with the other properties of $ f_{\star} $, we get
\beq
	\label{eq: p estimates}
	\lf\| f_{\star} \ri\|_{2}^2 = \OO( \mu - \theo), \qquad	\lf\| f_{\star} \ri\|_{4}^4 = \OO( (\mu - \theo)^2).
\eeq

Analogous results hold for the minimization of the functional $ \foneal $ on the interval $ [0,\ell] ${, i.e., 
\[
\foneal[f]= \int_0^{\ell} \diff t \: \lf\{ {f^{\prime}}^2 + (t+\alpha)^2 f^2 - \mu f^2 +\tfrac{1}{2} \mu f^4 \ri\},\qquad f\in H^1((0,\ell); \R)
\]} 
up to exponentially small errors in $ \ell $, which we denote as $ \OO(\ell^{-\infty}) $, i.e., quantities bounded by any arbitrary negative power of $ \ell $. This can be easily deduced as in the discussion contained in \cite[Sect. 3.1]{CR2} (see also \cite[Appendix A]{CR2} and \cite[Appendix A]{CG2}). However, for the sake of completeness, we state and prove the result in details below.

\begin{lemma}
	\label{lemma: 1D}
	For any $ \mu \in \lf(\mu_3, 1\ri) $, as $ \ell \to + \infty $,
	\beq
		\label{eq: l star estimate}
		\eonel(\mu) = \eones(\mu) + \OO(\ell^{-\infty}),	\qquad
		 \lf\| f_0 - f_{\star} \ri\|_{\infty} = \OO(\ell^{-\infty}).
	\eeq
\end{lemma}

\begin{proof}
	For given $ \alpha \in \R $, {let us denote by $ f_{\alpha}\in \domone $ and $ f_{\alpha, \ell}\in H^1((0,\ell); \R) $} the unique minimizers of $ \fonea $ and $ \foneal $, respectively. By exponential bounds analogous to the ones stated in \cite[Lemma 9]{CR2}, one can easily show that both $ f_{\alpha} $ and $ f_{\alpha,\ell} $ are smooth and decay exponentially, i.e.\footnote{{Here and in the following, $ C $ stands for a positive finite constant, whose value may vary from line to line and which is independent of the parameters of the model.}},
	\beq
		\label{eq: exp bounds}
		f_{\alpha}(t) \leq C e^{-\frac{1}{2}(t + \alpha)^2},	 	\qquad		f_{\alpha, \ell}(t) \leq C e^{-\frac{1}{2}(t + \alpha)^2},
	\eeq
	In fact, if $ \mu \leq \theo $, by the same argument mentioned above, $ f_{\alpha} \equiv 0 $ for any $ \alpha \in \R $, but the exponential bound hold nevertheless. Using $ f_{\alpha}  $ as a trial state for $ \foneal $ and exploiting this information, we immediately get the bound $ \eoneal \leq \eonea + \OO(\ell^{-\infty}) $. For the lower bound, it is then sufficient to construct an $ H^1$-regularization of $ f_{\alpha,\ell} $, which is positive and decreasing for $ t \geq \ell $, and evaluate $ \fonea $ on it. The exponential decay of $ f_{\alpha,\ell} $ then yields $ \eonea \leq \eoneal + \OO(\ell^{-\infty}) $. In conclusion, for any $ \alpha \in \R $,
	\beq
		\label{eq: enalpha diff}
		\eonea- \eoneal = \OO(\ell^{-\infty}).
	\eeq
	
	In order to deduce the result on the ground state energies $\eones$ and $\eonel$, we observe that, for $ \mu \leq \theo $, $  \eonea  = 0 $ for any $ \alpha \in \R $ and therefore the result easily follows. By analogous bounds to \eqref{eq: exp bounds} for the eigenstates of the linear problem, we also get
	\bdm
		\inf_{\alpha\in\R} \inf_{\lf\| f \ri\|_{L^2(0,\ell)}=1} \int_{0}^{\ell} \diff t \lf\{ |\partial_t f|^2 + (t+\alpha)^2 |f|^2 \ri\} = \theo + \OO(\ell^{-\infty}),
	\edm
	which immediately implies that, for $ \mu < \theo $, $ \lf\| f_0 \ri\|_2 = \OO(\ell^{-\infty}) $ and $ \lf\| f_0 \ri\|_4 = \OO(\ell^{-\infty}) $. This in turn can be used in the variational equation for $ f_0 $ to deduce the $ L^{\infty} $ bound. From
	\bdm
		- f_0^{\prime\prime} + (t + \alpha_0)^2 f_0 = \mu \lf(1 - f_0^2 \ri) f_0,
	\edm
	one gets that $ \lf\| f_0 \ri\|_{H^1(0,\ell)} = \OO(\ell^{-\infty}) $ and thus the estimate.
	
	 Exactly at the threshold $ \mu = \theo $, the proof is slightly more involved.  {First of all, we recall that $f_\star = 0$ and $E^{\mathrm{1D}}_\star(\Theta_0) = 0$, for $\mu = \Theta_0$ (see \eqref{eq: E* and f* at Theta0}). We therefore have to show that, for $\mu = \Theta_0$, $\|f_0\|_\infty = \mathcal{O}(\ell^{-\infty})$ and $E^{\mathrm{1D}}_\ell = \mathcal{O}(\ell^{-\infty})$. Using then the definitions of $\Theta_0$ and of $\mathcal{E}^{1D}_{\ell}$, one obtains that $ \lf\| f_0 \ri\|_4 = \OO(\ell^{-\infty})  $ and therefore the nonlinear term can be dropped from the functional $\mathcal{E}^{1D}_{\ell}$, up to corrections of order $ \OO(\ell^{-\infty}) $.} This in turn implies that $ f_0 = \lambda \psi_{0,\ell} + \OO(\ell^{-\infty}) $ in $ L^2$-sense, where $ \psi_{0,\ell} $ stands for the normalized minimizer of the linear problem and $ \lambda \in \R^+ $. However, plugging this ansatz in the energy, one gets that $ \frac{1}{2} \mu \lambda^4 \lf\| \psi_{0,\ell} \ri\|_4^4 = \OO(\ell^{-\infty}) $. Hence, since $ \lf\| \psi_{0,\ell} \ri\|_4 \geq c > 0 $, this yields $ \lambda = \OO(\ell^{-\infty}) $ and the proof is completed as above.

	On the opposite, for $ \mu > \theo $, we need to estimate the difference $ \alpha_0 - \alpha_{\star} $. This can be done  as in \cite[Prop. 1]{CR2}. First, we claim that, for any $ \alpha \in \R $,
\beq
	\label{eq: f estimate}
	\lf\| f_{\alpha,\ell} - f_{\alpha} \ri\|_{H^1(0,\ell)} =  \OO(\ell^{-\infty}),
\eeq
which in turn implies an analogous bound for the $ L^{\infty}$ norm of the difference. Next we use this as follows 
\bml{
	\label{eq: eoneal diff}
	\eoneal - \eones = {E}^{\mathrm{1D}}_{\alpha,\ell} - {E}^{\mathrm{1D}}_{\alpha_\star}= \int_0^{+\infty} \diff t \: \lf( 2 t ( \alpha - \alpha_{\star}) + \alpha^2 - \alpha^2_{\star} \ri) f^2_{\star} + \OO(\ell^{-\infty}) \\
	=  \lf( \alpha - \alpha_{\star} \ri)^2 \lf\| f_{\star} \ri\|_2^2 + \OO(\ell^{-\infty}),
}
where we used the optimality \eqref{eq: optimality}. Hence, since for any $ \mu > \theo $ there exists $ c > 0 $ independent of $ \ell $ such that $ \lf\| f_{\star} \ri\|_2^2 \geq c $, 
\bdm
	\eonel - \eones = \OO(\ell^{-\infty})
\edm
and 
\begin{equation}\label{eq: alpha0 - alpha*}
	 \alpha_{\star} - \alpha_0 = \OO(\ell^{-\infty}),
\end{equation}
 which yields the estimate of the difference between the minimizers via 
\bdm
	\lf\| f_{0} - f_{\star} \ri\|_{\infty} \leq \lf\| f_{\alpha_0, \ell} - f_{\alpha_{\star},\ell} \ri\|_{\infty}
	+ \lf\| f_{\alpha_{\star},\ell} - f_{\star} \ri\|_{\infty} = \OO(\ell^{-\infty}),
\edm
since, for any $ \alpha \in \R $, $  \lf\| f_{\alpha,\ell} - f_{0} \ri\|_{\infty} \leq C \lf| \alpha - \alpha_0 \ri| $. Indeed, as in \eqref{eq: eoneal diff}, one shows that, by the optimality of $ \alpha_0 $, 
\bdm
	 \lf| {E}^{\mathrm{1D}}_{\alpha,\ell} - \eonel \ri| \leq C \lf( \lf\| f_{\alpha,\ell} - f_{0} \ri\|_{H^1(0,\ell)} + \lf(\alpha - \alpha_0 \ri)^2 \ri),
\edm
which implies the estimate, as in the derivation of \eqref{eq: f estimate}.

It just remains to complete the proof of \eqref{eq: f estimate}. To this purpose, we set $ f_{\alpha, \ell} = : f_{\alpha} v $, for some $ v \in H^1(0,\ell) $ (recall that both $ f_{\alpha,\ell} $ and $ f_{\alpha} $ are real as so is $ v $). After some trivial calculation and by an integration by parts, we get the splitting 
\bdm
	\eoneal = \eonea + \int_{0}^{\ell} \diff t \: f_{\alpha}^2 \lf\{ {v^{\prime}}^2 + \tfrac{1}{2} \mu f_{\alpha}^2 \lf( 1 - v^2 \ri)^2 \ri\}
\edm
where we used the identity $ \eonea = - \frac{1}{2} \mu \lf\| f_{\alpha} \ri\|_4^4 $ and that, thanks to the Neumann conditions at $ 0 $ and $ \ell $, the boundary terms of the integration by parts vanish. Hence, by \eqref{eq: enalpha diff}, 
\bdm
	\lf\| f_{\alpha}^2 (1 - v^2) \ri\|_2 = \lf\| f_{\alpha,\ell}^2 - f_\alpha^2 \ri\|_2 = \OO(\ell^{-\infty}). 
\edm
Moreover, we get
\beq
	\label{eq: est 1}
	\OO(\ell^{-\infty}) = \int_{0}^{\ell} \diff t \: f_{\alpha}^2 {v^{\prime}}^2 = \int_{0}^{\ell} \diff t \: \lf( f_{\alpha,\ell}^{\prime} - v f_{\alpha}^{\prime} \ri)^2 \geq \lf( \lf\| f_{\alpha,\ell}^{\prime} -   f_{\alpha}^{\prime} \ri\|_2 - \lf\| f_{\alpha}^{\prime} (1 - v) \ri\|_2 \ri)^2.
\eeq
To bound the last term on the r.h.s., we act as follows:
\beq
	\label{eq: est 2}
	\lf\| f_{\alpha}^{\prime} (1 - v) \ri\|_2 \leq \lf\| f_{\alpha}^2 (1 - v^2) \ri\|_2 \lf\| \frac{f_{\alpha}^{\prime}}{f_{\alpha}} \ri\|_4^2 = \OO(\ell^{-\infty}),
\eeq
since 
\bdm
	\lf\| \frac{f_{\alpha}^{\prime}}{f_{\alpha}} \ri\|_{L^4(\R^+)} \leq C.
\edm
Indeed, setting $ g_{\alpha} : = f_{\alpha}^{\prime}/f_{\alpha} $, the variational equation for $ f_{\alpha} $ yields
\bdm
	- g_{\alpha}^{\prime} + g_{\alpha}^2 + (t + \alpha)^2 - \mu (1 - f_{\alpha}^2) = 0,
\edm
and multiplying by $ g_{\alpha}^2 $ and integrating, we get (recall that $ g_{\alpha}(0) = 0 $)
\bdm
	\lf\| g_{\alpha} \ri\|_{L^4(\R^+)}^4 + \lf\| (t + \alpha) g_{\alpha} \ri\|_{L^2(\R^+)}^2 \leq \mu \lf\| g_{\alpha} \ri\|_{L^2(\R^+)}^2,
\edm
which by interpolation implies the boundedness of $ \lf\| g_{\alpha} \ri\|_{L^4(\R^+)}^4 $. In conclusion, plugging \eqref{eq: est 2} in \eqref{eq: est 1}, we get 
\bdm
	\lf\| f_{\alpha,\ell}^{\prime} -   f_{\alpha}^{\prime} \ri\|_2 = \OO(\ell^{-\infty})
\edm
and therefore \eqref{eq: f estimate}. 
\end{proof}

\subsection{Magnetic Laplacian in a sector}\label{subsect: magnetic lapl}

An important role in our analysis is played by the magnetic Laplacian. We thus denote by \begin{equation}\label{eq: def Lbeta}
\Lap_{\beta} =-(\nabla-\ii\Fb)^2
\end{equation}
 the self-adjoint realization in $L^2(\S_{\beta})$ of the magnetic Laplacian with uniform magnetic field and Neumann boundary conditions \( \nv \cdot(\nabla-\ii\Fb)u=0\) on  \(\partial \S_{\beta}\) (see \eqref{eq: sector}), where \(\nv\) stands for the unit inward normal to \(\partial\S_{\beta} \). We recall that the bottom of the spectrum of $ \Lap_\beta $ is given by \eqref{eq: mu beta} and denoted by $ \mu_{\beta} $.   By Persson's lemma, the bottom of the essential spectrum of  $ \Lap_\beta $ is independent of the angle $ \beta $ (see \cite[Lemma 2.2]{Bo})  
\begin{equation}\label{eq:Theta0}
\inf\sigma_{\rm ess}(\Lap_\beta)=\inf\sigma(\Lap_\pi) = : \Theta_0.
\end{equation}
Let us recall two inequalities related to the spectrum of the  magnetic Laplacian in the half-plane and the plane, respectively (see, e.g., \cite[Prop. 6.1]{HM}; the version including the vertex is stated in \cite[Thm 1.2]{Bo}). The idea is that if the support of the function may reach the boundary of the sector, excluding the vertex, then, after blow-up, the effective spectral problem is given by the magnetic Laplacian on the half-plane. On the contrary, if the support of the function is contained in the interior of $ \S_{\beta} $, no boundary is present in the effective model which becomes the magnetic Laplacian in the whole of $ \R^2 $. From the spectral point of view such effective models have different ground state energies given by $ \theo $ in the half-plane and $ 1 $ in the plane, respectively.

Given $R>0$,   let  $\cB_R$ denote the open disc centered at  the vertex  $\mathbf{0}$ of $\S_\beta$,  and  of radius $R$.  We deal with two types of domains:
\begin{itemize}
\item {{\it inner}} domains {contain the vertex of the sector} and are of the form $\S_\beta\cap \cB_R$;
\item {{\it outer}} domains are on the opposite far from the corner and are of the form $\S_\beta\setminus\overline{\cB}_R$.
\end{itemize}
\begin{lemma}\label{lem:spectrum}
Let $\Omega \subset \S_{\beta} $ be an open set and let \( \psi \in W^{1,2}_{\Fb}(\S_\beta)\). 
\begin{itemize}
\item If $\Omega$ is {an innner domain} and $\psi=0$ on $\S_\beta\cap\partial \Omega$,  then
\beq
	\int_{\Omega} \diff \xv \: \lf| \lf(\nabla+\ii\Fb \ri) \psi \ri|^2 \geq \mu_\beta \int_{\Omega} \diff \xv \: \lf| \psi \ri|^2;
\eeq
\item If $\Omega$ is {an outer domain} and $\psi=0$ on $\S_\beta\cap\partial \Omega$,  then 
\beq
	\int_{\Omega} \diff \xv \: \lf| \lf(\nabla+\ii\Fb \ri) \psi \ri|^2 \geq \theo \int_{\Omega} \diff \xv \: \lf| \psi \ri|^2;
\eeq
\item 
If  $\Omega$ is {an outer domain} and ${\rm supp\,} \psi \subset \Omega  $,  then 
\beq
	\int_{\Omega} \diff \xv \: \lf| \lf(\nabla+\ii\Fb \ri) \psi \ri|^2 \geq \int_{\Omega} \diff \xv \: \lf| \psi \ri|^2.
\eeq
\end{itemize}
\end{lemma}

\subsection{{Another corner effective energy at the threshold $\mu = \Theta_0$}}\label{subsec: Dirichlet eff energy}\label{sec: dirichlet corner var prob}

In addition to the variational problems introduced in \eqref{eq:gse-corner} and \eqref{eq-E-CG},  which are studied in details in \cite[Sect. 3]{CG2},  we introduce another one on $ \corner $.  We set  for $ \mu \in [\theo, 1) $
\begin{equation}\label{eq:gse-corner-D}
 	\ecordl(\mu) = \inf_{\psi\in \domcord} \E_\mu[\psi; \corner]
\end{equation}
where
\beq\label{eq: def domain Dirichlet corner}
	\domcord = \{\psi\in H^1(\R^2) \: \big| \: \psi=0 \mbox{ on }\partial\Gamma^{\rm bd} \cup \partial\Gamma^{\rm in}  \},
\eeq
with minimizer $ \mdirl $.
Such an energy is expected to approximate $ \ecor (\mu)$ when $ \mu \to \theo^+ $ (see \cref{prop:mu=Th0}): indeed, on one hand, $ \eonel(\theo) = o_{\ell}(1)$, so that the first term in \eqref{eq-E-CG} can be dropped, and, on the other one,   
 the boundary conditions in $ \domcor $ are given in terms of a function $ \psi_{0} $ such that (see \cref{lemma: 1D}),
\bdm
	\psi_{0} \xrightarrow[\mu \to \theo]{\mathrm{pointwise}} 0.
\edm
The main reason behind the introduction of such a variant of $ \ecorl $ is that the former, thanks to the Dirichlet conditions at the boundaries $ \partial\Gamma^{\rm bd} \cup \partial\Gamma^{\rm in} $, can be linked more easily to a variational problem in the full sector $ \S_{\beta} $.

For some given $a > 1 $, we introduce the set of parameters
\beq
	\mathscr A=\lf\{(L,\ell)\in\R_+ \times \R_+ \: \big| \: 1\leq \ell<L\tan\tfrac{\beta}{2} \leq C \ell^{a} \ri\}.
\eeq
Note that, by construction, if one of the two parameters in $ \mathscr{A} $ tends to $ + \infty $, the other is forced to do the same.  Then, we define the following quantity
\begin{equation}\label{eq:def-en-th0}
\edir(\theo)= \inf_{(L,\ell)\in\mathscr A}\ecordl(\theo),
\end{equation}
which we will prove to be finite and non-positive.  In fact,  the analogue of \cite[Prop. 3.4]{CG2} for $ \ecordl(\theo) $ is the next statement.  Note indeed that the former result is formulated for $ \mu > \theo $ and it is not immediately obvious whether it extends at the threshold.

\begin{lemma} \label{lemma:mu=Th0}
Let \eqref{eq: condition} hold.   Then,
\beq\label{eq:lim-mu=theo}
	 \lim_{\substack{\ell\to+\infty\\ \ell\ll L\ll \ell^a}} \ecordl(\theo) {= \edir(\theo)}
\eeq
and  $-\infty<\edir(\theo)\leq0$.  
\end{lemma}
\begin{proof}
Thanks to the Dirichlet conditions at the boundary, $ \mathscr{D}^{\mathrm{D}}_{L,\ell_1} \subset \mathscr{D}^{\mathrm{D}}_{L,\ell_2} $ for any $ \ell_1 \leq \ell_2<L\tan\tfrac{\beta}{2} $. Hence, the function \(\ell\mapsto \ecordl(\mu) \) is monotone non-increasing on $(0,L\tan\tfrac{\beta}{2})$ for fixed $ L$  and $\mu\in[\theo,+\infty)$.  Similarly, \(L \mapsto \ecordl (\mu)\) is non-increasing on $(\ell/\tan\tfrac{\beta}{2},+\infty)$.  Consequently,  for a fixed $\ell\geq 1$,  there exists
\[
	\ecord(\mu) : = \lim_{L \to+\infty} \ecordl(\mu)  = \inf_{L {\in \R_+ \: | \: (\ell,L) \in \mathscr{A}}} \ecordl(\mu) .
\]
The function \( \ell \mapsto  \ecord(\mu) \) is then monotone non-increasing and therefore the following quantity is well-posed:
\[ 
	\edir(\mu)  : =  \lim_{\ell \to+\infty} \ecord(\mu)  = \inf_{\ell \in[1,+\infty)} \ecord(\mu)  = \lim_{\ell \to+\infty}  \lim_{L \to+\infty} \ecordl(\mu) . 
\]
For $\mu=\theo$,  the existence of the limit in \eqref{eq:lim-mu=theo} follows and coincides with the above quantity (and with \eqref{eq:def-en-th0}),    since the limit is taken on a subsequence.

By the variational equation for the minimizer,  we get that for any $ (\ell, L) \in \mathscr{A} $, \( \ecordl(\Theta_0) \leq 0\),   which implies $\edir(\theo)\leq 0$ by \eqref{eq:def-en-th0}.
Next, to show that $\ecor(\theo)$ is finite,   we prove that there exists a positive constant \(C  \), such that
\beq\label{eq:claim-e-finite}
\ecordl(\theo)\geq -C,		\qquad		\forall L> 2\mbox{ and }2<\ell<L\tan\tfrac{\beta}{2}.
\eeq

 Let $ \mdirl \in \domcord $ be a minimizer realizing $ \ecordl(\theo) $. We recall the bound \eqref{eq: infty}, i.e., \(|\mdirl|\leq 1\). Consider a partition of unity such that
\[\chi_1^2+\chi_2^2=1\mbox{ on }\overline{\Gamma_\beta(L,\ell)},\qquad \begin{cases}
\chi_1=1\mbox{ on }\overline{\Gamma_\beta(L,\ell)\cap\cB_1}\\
\chi_2=1\mbox{ on }\overline{\Gamma_\beta(L,\ell)\setminus\cB_2}
\end{cases}\]
{and $ \|\nabla\chi_1\|_\infty + \|\nabla\chi_2\|_\infty \leq C $.}
Then,  we can write
\begin{multline}
	\label{eq: partition ecordl}
	\ecordl(\theo) 
	\geq \E_{\theo} \big[\chi_1\mdirl; \corner\big] + \E_{\theo} \big[\chi_2\mdirl; \corner\big] -\sum_{i=1}^2\int_{\Gamma_\beta(L,\ell)} \diff \xv |\nabla\chi_i|^2 |\mdirl|^2	\\
	\geq \E_{\theo} \big[\chi_1\mdirl; \corner\big] + \E_{\theo} \big[\chi_2\mdirl; \corner\big] - C
\end{multline}
since \(|\mdirl|\leq 1\) and \(\nabla\chi_1,\nabla\chi_2\) are supported in \(\cB_2\setminus\cB_1\).  Furthermore, since \(\chi_1\) is supported in \(\overline{\Gamma_\beta(L,\ell)\cap \cB}_2\), by \cref{lem:spectrum}, we get that
\[
	\E_{\theo}\big[\chi_1\mdirl; \corner\big] \geq (\mu_\beta-\Theta_0)\int_{\Gamma_\beta(L,\ell)\cap\cB_2 } \diff \xv \: \big|\mdirl \big|^2 \geq (\mu_\beta-\Theta_0)|\cB_2  |,
\]
where we used \eqref{eq: condition} and \(|\mdirl|\leq 1\).  Finally, again by \cref{lem:spectrum},
\[
	\E_{\theo}\big[\chi_2\mdirl; \corner\big] \geq 0.
\]
Plugging the above inequalities in \eqref{eq: partition ecordl}, we get $ \ecordl(\theo)  \geq -C +4(\mu_\beta-\Theta_0)\pi $, and then the claim in \eqref{eq:claim-e-finite}.
\end{proof}

\section{Corner effective model}\label{sec:cor}

We investigate here the extension of the corner effective models to the threshold $ \mu = \theo $. An important role will be played by the models with Dirichlet conditions introduced in \eqref{eq:gse-corner-D}.

\subsection{{Corner energy at the threshold $\mu = \Theta_0$}}

The importance of the variational problem introduced in \eqref{eq:gse-corner-D} relies on the following result.

\begin{proposition}[Value of $ \ecor $ at $ \theo $] \label{prop:mu=Th0}
\mbox{}	\\
Let \eqref{eq: condition} hold.  Then, 
\beq
	\label{eq: limit ecord}
	\lim_{\mu \to \theo^+} \ecor(\mu) =  \edir(\theo).
\eeq
\end{proposition}
The proof of \cref{prop:mu=Th0} occupies the rest of this section.   The first property we state  is about the  decay in the normal direction {of any minimizer of the} variational problem introduced in \eqref{eq:gse-corner},  which is going to play a key role in the rest of the discussion (see \cite[Lemma B.3]{CG2} for the proof).   Note that the decay is in this case w.r.t. the distance from the outer boundary {(i.e., the distance $t(\xv)$ introduced in \eqref{eq:def-dist-t})} rather {than from the corner}.

Recall that $ \mcor $ stands for any minimizer of the energy functional $ \E_\mu[ \, \cdot \,; \corner] $ (see \eqref{eq:gse-corner}) on the domain $ \domcor $.

\begin{proposition}[Agmon decay of $ \mcor $]
\label{prop:decay corner}
\mbox{}	\\
Let $ 1 \ll \ell \leq L \ll \ell^{a} $ {as $\ell \rightarrow \infty$}. Then, {$ \exists \eps_0 > 0 $, such that,} for any minimizer $ \mcor $  and for any $ \mu \in [\theo, 1) $,
\beq
	\label{eq: decay corner}
	\int_{\corner}  \diff \xv \: \lf\{ \lf| \lf( \nabla + \ii \fv \ri) \mcor \ri|^2 + \lf| \mcor(\xv) \ri|^2 \ri\} \ee^{\varepsilon_0\sqrt{1 - \mu}\, \mathrm{dist}(\xv, \partial \Gamma^{\mathrm{out}})} \leq {C}L.
\eeq
\end{proposition}

\subsection{An effective model on the wedge}
In order to prove {\eqref{eq: limit ecord}}, we need to introduce yet another minimization problem in a variant of the corner region $ \corner $, namely a wedge $ \wed $ where the interior boundary is removed. More precisely, we set
\beq
	\wed : = \lf\{ \xv \in \mathcal{S}_{\beta} \: \big| \:  0 < |s(\xv)| < L \ri\}.
\eeq
Note that {in the definition of $\wed$,} we are committing a little abuse {of notation, since tubular coordinates are not well-posed on the bisectrix of $ \S_{\beta} $. However, with the sole exception of the bisectrix, the above definition makes perfect sense.}
For $  L \gg 1 $, let $ g \in C^{\infty}(\R^+) $ be such that 
\beq\label{eq: def gt}
	g(t) = 
	\begin{cases}
		f_\star(t) 	&	\mbox{for } t \in \lf[0, \tfrac{1}{3} L \tan \tfrac{\beta}{2}\ri],	\\
		0		&	\mbox{for } t \geq \tfrac{2}{3} L \tan \tfrac{\beta}{2},
	\end{cases}
\eeq
and it is monotonically decreasing for $ t \in  [\tfrac{1}{3} L \tan \tfrac{\beta}{2}, \tfrac{2}{3} L \tan \tfrac{\beta}{2} ] $. Note that $ g $ is different from $ f_\star $ only where the latter is exponentially small in $ L $. Then, we consider the following energy
\beq
	\label{eq: ewedl}
	\ewedl(\mu) = \inf_{\psi \in  \domwed} \E_{\mu}[\psi; \wed],
\eeq
where
\beq
	\label{eq: domwed}
	\domwed = \lf\{ \psi \in H^1(\wed) \: \big| \: \psi=\psi_{g} \mbox{ on }\partial W^{\rm bd} \ri\},
\eeq
$ \partial W^{\rm bd} : = \overline{\wed} \cap \lf\{ s = \pm L \ri\} $ and
\beq
	\label{eq: psig}
	\psi_{g}(\xv(s,t)) = g(t) e^{\ii\alpha_\star s -\frac12 \ii st }.
\eeq
Note again the little abuse of notation, since $ t $ has a different meaning depending on the value of $ s $. {However, we do not see the discontinuity of the phase at the bisectrix thanks to the fact that $g(t) = 0$ for $t\geq (2/3) \tan (\beta/2)$.} Furthermore, thanks to its monotone behavior for large $ t $, $ \psi_g $ satisfies the same decay estimates \eqref{eq: exp bounds} as {$ f_0 $ (see \cref{lemma: 1D})} .

We denote by $\psi^{\mathrm{wedge}}_{\mu}$ any minimizer associated to \eqref{eq: ewedl}. We also stress that this new variational problem depends on a single parameter $ L > 0 $. Thanks to the exponential decay in the distance from the outer boundary of both $ \mcor $ and $ f_0 $, it is however easy to prove the following result.

\begin{lemma}
\label{lemma: wedge vs corner}
\mbox{}	\\
Let $\ell \ll L \ll \ell^a$ as $\ell\rightarrow \infty$. For any $\Theta_0 < \mu < 1$, 
\beq
	\label{eq: wedge and corner en}
    \ewedl(\mu) - \ecorl(\mu) = \OO(\ell^{-\infty}).
\eeq
Furthermore, for any $ 1 \ll L_1 \leq L_2 $, 
\beq
	\label{eq: monotonicity}
	E^{\mathrm{wedge}}_{L_1} (\mu)- 2 \eones(\mu) L_1 \geq E^{\mathrm{wedge}}_{L_2}(\mu) - 2 \eones(\mu) L_2 + \OO(L_1^{-\infty}),
\eeq
and
\beq
	\label{eq: alternative ecor}
	\ecor(\mu) = \lim_{L \to + \infty} \lf( \ewedl(\mu) - 2 \eones(\mu) L \ri).
\eeq
\end{lemma}
\begin{proof}
	We first point out that $ \mwed $ satisfies an exponential decay w.r.t. the distance from the outer boundary $ \partial W^{\rm bd}$,  which is completely analogous to the one stated in \cref{prop:decay corner}, i.e., for $ \eps_0 \in (0,1) $, $ \mu < 1 $ and $ L \gg 1 $,
	\beq
		\label{eq: decay wedge}
		\int_{\wed}  \diff \xv \: \lf\{ \lf| \lf( \nabla + \ii \fv \ri) \mwed \ri|^2 + \lf| \mwed(\xv) \ri|^2 \ri\}  \ee^{\eps_0\sqrt{1 - \mu}\, \mathrm{dist}(\xv, \partial \Gamma^{\mathrm{out}})} \leq CL,
	\eeq
where $C$ is independent of $\mu$ if we take {$\mu\in(0,1-\eps]$} with a fixed $\eps\in(\theo,1)$. 	
	Furthermore, by \eqref{eq: l star estimate}, the above decay estimate and \eqref{eq: decay corner}, we can easily replace $ \psi_0 $ with $ \psi_g $ in the boundary condition on $ \partial \Gamma^{\mathrm{bd}} $. We omit the proof for the sake of brevity. Altogether these facts immediately imply \eqref{eq: wedge and corner en}.

	Concerning \eqref{eq: monotonicity}, we evaluate the energy in $ W_{\beta}(L_2) $ on the following trial state:
	\bdm
		\psi_{\mathrm{trial}}(\xv) = 
		\begin{cases}
			\psi^{\mathrm{wedge}}_{\mu, L_1}	&	\mbox{for } |s| \leq L_1		\\
			{\phi_{L_1,L_2}}								&	\mbox{for } |s| \in (L_1,L_2],
		\end{cases}
	\edm
	{where we keep explicitly track of the dependence on $ L $ by setting $ \psi^{\mathrm{wedge}}_{\mu,L} := \psi^{\mathrm{wedge}}_{\mu} $. Note that the function $ g $ in \eqref{eq: def gt} depends itself on $ L $ and therefore we set $ g_L : = g $.}
	The function {$ \phi_{L_1,L_2} $ is a smooth function obtained by gluing together two functions of the form \eqref{eq: psig}:}
\beq
	{\phi_{L_1,L_2} = \lf( \tfrac{L_2 - s}{L_2 - L_1} g_{L_1}(t) + \tfrac{s - L_1}{L_2 - L_1} g_{L_2}(t) \ri)  e^{\ii\alpha_\star s -\frac12 \ii st }.}
\eeq
By plugging $ \psi_{\mathrm{trial}} $ in $ \E_{\mu}[\psi; W_{\beta}(L_2)] $, we immediately get
	\beq
		E^{\mathrm{wedge}}_{L_2}\leq \mathcal{E}_\mu[\psi_{\mathrm{trial}}, W_\beta(L_2)] =  E^{\mathrm{wedge}}_{L_1} + 2 \E_{\mu} [ \psi_{\mathrm{trial}};  W_{\beta}(L_2) \setminus W_{\beta}(L_1)],
	\eeq
	where we used the definition of $\psi_{\mathrm{trial}}$. Hence{, recalling that $ g_{L_1}(t) = g_{L_2}(t) = f_{\star}(t) $ for any $ t \in [0,\frac{1}{3} L_1  \tan \frac{\beta}{2}] $, we get}
	\bmln{
		\E_{\mu} [ \psi_{\mathrm{trial}};  W_{\beta}(L_2) \setminus W_{\beta}(L_1)] = 2 \int_{L_1}^{L_2} \diff s \int_0^{{\frac{1}{3} L_1 \tan \frac{\beta}{2}}} \diff t \: \lf\{ {f_\star^{\prime}}^2 + (t + \alpha_\star)^2 f_\star^2 - \mu f_\star^2 + \tfrac{1}{2} \mu f_\star^4 \ri\} \\
		+\OO(L_1^{-\infty})
		= 2 \eones (L_2 - L_1) + \OO(L_1^{-\infty}),
	}
	{so that} \eqref{eq: monotonicity} is proven. Therefore the quantity in the right hand side of \eqref{eq: alternative ecor} is well defined and it exists. Recalling that 
	\[
		\ecor(\mu) = \lim_{\substack{\ell\to+\infty\\ \ell\ll L\ll \ell^a}} \lf( -2L \eonel(\mu) + E^{\mathrm{corn}}_{L,\ell}(\mu) \ri),
	\]
	the proof of \eqref{eq: alternative ecor} is then a consequence of \eqref{eq: wedge and corner en} and of \cref{lemma: 1D}.
\end{proof}

We are now in position to prove \cref{prop:mu=Th0}.

\begin{proof}[Proof of \cref{prop:mu=Th0}]
    To prove the result, we start by showing that the upper bound below holds
\beq
	\label{eq: up at theo}
   \ecor(\mu) \leq \lim_{\substack{\ell\to+\infty\\ \ell\ll L\ll \ell^a}} \ecordl(\theo) + C(\Theta_0 - \mu).
\eeq
 From the definition of $ \ecor $ and \eqref{eq: alternative ecor} in \cref{lemma: wedge vs corner}, we get
\bdm
	\ecor(\mu) = \lim_{L \to + \infty} \lf( \ewedl(\mu) - 2 \eones (\mu)L\ri),
\edm
but, by the monotonicity of the quantity on r.h.s. stated again in \cref{lemma: wedge vs corner}, we also have that for any finite $ L_{\mu} \gg 1 $
\beq
	 \lim_{L \to + \infty} \lf( \ewedl(\mu) - 2 \eones(\mu) L \ri) \leq E^{\mathrm{wedge}}_{L_{\mu}}(\mu) - 2 \eones(\mu) L_{\mu} + o_{L_\mu}(1),
\eeq
where $ L_{\mu} $ is going to be chosen depending of $ \mu $ in such a way that 
\bdm
	L_{\mu} \xrightarrow[\mu \to \theo^+]{} +\infty.
\edm

Next, let us consider a variant of the minimization \eqref{eq: ewedl} exactly at the threshold $ \mu = \theo $ and with the boundary conditions on $ \partial W_{\beta}(L) $ replaced with zero conditions:
\beq
	\label{eq: ewedd}
	E^{\mathrm{D}, \mathrm{wedge}}_{L} (\theo) = \inf_{\psi \in \mathscr{D}_L^{\mathrm{D}}} \E_{\theo}[\psi; \wed],
\eeq
where
\beq
	\label{eq: domwedd}
	\mathscr{D}_L^{\mathrm{D}, \mathrm{wedge}} = \lf\{ \psi \in H^1(\wed) \: \big| \: \psi=0 \mbox{ on }\partial W^{\rm bd} \ri\}.
\eeq
We denote $ \psi^{\mathrm{D}}_L $ any corresponding minimizer, which can be easily shown to satisfy the same exponential decay as in \eqref{eq: decay wedge}. {Furthermore, using such a decay, one can easily show that
\beq
	\label{eq: wedges at threshold}
	E^{\mathrm{D}, \mathrm{wedge}}_{L}(\theo) - E^{\mathrm{D}}_{L, \ell}(\theo) = \OO(\ell^{-\infty}),
\eeq
for any $ (L,\ell) \in \mathscr{A}  $.}

{We now use  $ \psi^{\mathrm{D}}_L $ to construct a trial state to bound from above $ E^{\mathrm{wedge}}_{L_{\mu}} $: let $ 0 < \delta < L_{\mu} $ and
\beq
	\psi_{\mathrm{trial}}  =
	\begin{cases}
		\psi^{\mathrm{D}}_{L_{\mu} - \delta}	&	\mbox{for } |s| \leq L_{\mu} - \delta	\\
		 \phi 													&	\mbox{for } L_{\mu} - \delta \leq |s| \leq L_{\mu},
	\end{cases}
\eeq
where 
\bdm
	\lf. \phi \ri|_{s = L_{\mu} - \delta} = 0,		\qquad	\lf. \phi \ri|_{s = L_{\mu}} = \psi_g,
\edm
and $ \phi $ is smooth everywhere. In particular, we can require that
\beq
	\label{eq: delta est}
	\lf\| \nabla \psi_{\mathrm{trial}} \ri\|_{L^2(W_{\beta}(L_{\mu}) \setminus W_{\beta}(L_{\mu}-\delta))}^2 \leq C (\delta^{-1} + \delta) \lf\| f_\star\ri\|_{L^2(\R^+)}^2.
\eeq
Then, $ \psi_{\mathrm{trial}} \in \mathscr{D}^{\mathrm{wedge}}_{L_{\mu}} $ and 
\bml{
	E^{\mathrm{wedge}}_{L_{\mu}}(\mu) \leq \E_{\theo}\big[\psi^{\mathrm{D}}_{L_{\mu} - \delta}; W_{\beta}(L_{\mu}-\delta)\big] + \tfrac{1}{2} (\mu - \theo) \big\| \psi^{\mathrm{D}}_{L_{\mu} - \delta} \big\|_{L^4(W_{\beta}(L_{\mu} -  \delta))}^4 \\
	+  \E_{\mu}\big[\psi_{\mathrm{trial}}; W_{\beta}(L_{\mu}) \setminus W_{\beta}(L_{\mu}- \delta) \big]  \leq E^{\mathrm{D}, \mathrm{wedge}}_{L_{\mu} - \delta} (\theo) + C (\mu - \theo) L_{\mu}  \\
	+ C  \lf( \delta\lf\| f_\star \ri\|_{L^2(\R^+)}^2 +\delta \lf\| f_\star \ri\|_{L^4(\R^+)}^4  + \delta^{-1} \lf\| f_\star \ri\|_{L^2(\R^+)}^2 \ri)  + \OO(L_{\mu}^{-\infty}),
	}
	where we used \eqref{eq: delta est}. Plugging \eqref{eq: p estimates} in the bound above and using \eqref{eq: asympt mu theo}, we get
\bml{
	E^{\mathrm{wedge}}_{L_{\mu}}(\mu) - 2 \eones(\mu) L_{\mu} \leq E^{\mathrm{D}, \mathrm{wedge}}_{L_{\mu} - \delta} (\theo) + C \lf[ (\mu - \theo) \lf(L_{\mu} + \delta  + \delta^{-1} \ri)  +  (\mu - \theo)^2 L_{\mu} \ri]	\\
	+ o_{L_\mu}(1).
}}
By choosing, e.g., 
\beq
	L_\mu = (\mu - \theo)^{-\frac{1}{2}},	\qquad		\delta = L_\mu = (\mu - \theo)^{\frac{1}{3}}.
\eeq
{and taking the superior limit $ \mu \to \theo^+ $ of the r.h.s., we get} 
\bdm
	{\limsup_{\mu \to  \theo^+} \ecor(\mu)  \leq \edir(\theo),}
\edm
{thanks to \eqref{eq: wedges at threshold}.}

To prove the corresponding lower bound, we can proceed similarly: let $ L > 0 $ and consider the following trial state in $ \mathscr{D}_{L+\delta}^{\mathrm{D}} $, for some $ \delta > 0 $:
\[
	{\psi_{\mathrm{trial}}:= 
	\begin{cases}
		\chi(s) \psi_{g}(\xv(-L,t))	&	\mbox{for } s \in [-L-\delta,-L]	\\
		\psi^{\mathrm{wedge}}_{\mu}	&	\mbox{for } |s| \leq L,	\\
		\chi(s) \psi_{g}(\xv(L,t))	&	\mbox{for } s \in [L,L+\delta],
	\end{cases}}
\]
{for some smooth and monotonically decreasing $ \chi(s) = \chi(|s|): [L, L+\delta] \mapsto [0,1] $ such that $ \chi(L) = 1, \chi(L+\delta) = 0 $ and}
\[
	\lf\|\nabla\psi_{\mathrm{trial}} \ri\|_{L^2(W_\beta(L + \delta) \setminus W_\beta(L))} \leq C(\delta^{-1} + \delta)\|f_\star\|^2_{L^2(\mathbb{R}^+)}.
\]
{Inserting such a trial state into $ \mathcal{E}_{\mu}[\, \cdot \,, W_{\beta}(L+\delta)] $, we obtain
\bml{
	E^D_{L + \delta}(\Theta_0) \leq \mathcal{E}_{\mu}[\psi^{\mathrm{wedge}}_{\mu}; W_{\beta}(L)] + \tfrac{1}{2}(\mu-\Theta_0)\big\| \psi^{\mathrm{wedge}}_{L}\big\|^2_{L^2(W_\beta(L))} \\
	+ \mathcal{E}_{\Theta_0} [\psi_{\mathrm{trial}}; W_{\beta}(L + \delta)\setminus W_\beta(L)].
}
Futhermore, 
\begin{multline*}
	\mathcal{E}_{\Theta_0}[\psi_{\mathrm{trial}}; W_\beta(L+\delta)\setminus W_\beta(L)] \leq \mathcal{E}_{\mu}[\psi_{\mathrm{trial}}; W_\beta(L+\delta)\setminus W_\beta(L)] + C(\mu-\Theta_0)(L+\delta) \\
	 \leq C \lf[  \delta \lf\| f_\star \ri\|_{L^2(\R^+)}^2 + \delta \lf\| f_\star \ri\|_{L^4(\R^+)}^4 + \delta^{-1} \lf\| f_\star \ri\|_{L^2(\R^+)}^2 + (\mu-\Theta_0)(L+\delta) \ri] + \OO(L^{-\infty})
\end{multline*}
and
\[
	 \mathcal{E}_{\mu}[\psi^{\mathrm{wedge}}_{\mu}; W_{\beta}(L)] = E^{\mathrm{wedge}}_{L} (\mu) \leq   -2L_\mu E^{\mathrm{1D}}_\star + E^{\mathrm{corn}}_{L_,\ell}(\mu) + \mathcal{O}(\ell^{-\infty}),
\]
for any $ (L,\ell) \in \mathscr{A} $, where we used that $-2L E^{\mathrm{1D}}_\star \geq 0$ and \cref{lemma: wedge vs corner}. We can now take as before
\bdm
	L = L_{\mu} \xrightarrow[\mu \to \theo^+]{} + \infty,		\qquad	\mbox{such that } (\mu-\Theta_0)L_{\mu} \xrightarrow[\mu \to \theo^+]{} 0
\edm
and $ \delta = C $, and exploit \eqref{eq: p estimates}, to get 
\bdm
	\limsup_{\mu \to \theo^+} E_{L_{\mu}}^{\mathrm{D}, \mathrm{wedge}}(\Theta_0)   \leq  \ecor(\theo),
\edm
which implies the result via \eqref{eq: wedges at threshold}.}
\end{proof}

\section{Sector effective model}\label{sec:sec}
The purpose of this section is to relate the sector effective energy with the energy at the threshold as stated in Theorem~\ref{theo: main}.

\subsection{Basic properties}

The key role of the condition \eqref{eq: condition} is to ensure that the variational problem \eqref{eq:ref-en*} admits a non-trivial minimizer.  

We collect several useful properties of such a minimization problem in the next proposition. Although some of the properties were already known (see, e.g., in \cite[Prop. 6.1]{BNF}), we spell the proof in details for the sake of the reader's convenience. 

\begin{proposition}[Minimization of $ \fsec $]
\label{prop:min} 
\mbox{}	\\
Let \eqref{eq: condition} hold. Then, 
\begin{enumerate}
\item the function \(\mu\mapsto \esec(\mu) \) is decreasing and concave on the interval \((\mu_{{\beta}},\Theta_0)\);
\item  for any $ \mu_\beta < \mu \leq \theo $, $ \esec(\mu) < 0 $, and, for $ \mu < \theo $,  there exists a non-trivial minimizer $ \msec \in W^{1,2}_{\Fb}(\R^2)\), such that
\beq
	\label{eq: msec bound 1}
	\int_{\S_\beta} \diff \xv \: \lf\{ \lf| \lf(\nabla + \ii\Fb \ri) \msec \ri|^2+ \lf|\msec \ri|^2 + \lf| \msec \ri|^4 \ri\} \leq \tfrac{C}{\Theta_0-\mu},
\eeq
for some $ C $ independent of $ \mu $.
\end{enumerate}
\end{proposition}

\begin{remark}[Boundedness from below at $ \theo $]
\mbox{}	\\
{As it follows from its proof, the statement of \cref{prop:min} does not directly imply that \( \esec(\theo) \) is finite at \(\mu=\Theta_0\), since the functional \( \fsec \) is not necessarily bounded from below there.} That $\esec(\theo)$ is finite will in fact be proved later on in \cref{prop:sb-theta0}.
\end{remark}

\begin{proof} If \eqref{eq: condition} holds, $ \mu_{\beta} $ is an isolated eigenvalue of {$ \Lap_{\beta} = -(\nabla - i\mathbf{F})^2  $ with Neumann boundary condition on $\mathcal{S}_\beta$ (see \eqref{eq: def Lbeta} for a precise definition)} and therefore there exists a corresponding normalized eigenstate $ u_{\beta} \in W^{1,2}_{\Fb}(\R^2) $. Using $ \lambda u_{\beta} $, $ \lambda > 0 $, as a trial state for {$ \fsec[\cdot] = \mathcal{E}_\mu[\cdot; \mathcal{S}_\beta]$}, we get for any $ \mu_\beta < \mu \leq \theo $
\beq\label{eq: Esector negative}
	\esec(\mu) \leq \lambda^2 \lf( \mu_{\beta} - \mu + \tfrac{1}{2} \lambda^2 \lf\| u_{\beta} \ri\|_4^4 \ri) < 0,
\eeq
provided $ \lambda $ is taken small enough. This immediately implies that any minimizer, if it exists, must be non-trivial since the ground state energy is non-zero.

Next, we investigate the properties of the function \(\mu\mapsto \esec(\mu)\). Consider the functional {$ \mathcal{F}_{b}[\psi] : = \mu^{-1}\mathcal{E}_{\mu}[\psi, \mathcal{S}_\beta]$}, with $ b = \mu^{-1} $. Its explicit expression is
\bdm
	 \mathcal{F}_{b}[\psi] = \int_{\mathcal{S}_{\beta}} \diff \xv \: \left[ b \lf| \lf(\nabla + \ii \fv \ri) \psi \ri|^2 -  |\psi|^2 + \tfrac{1}{2}  |\psi|^4\right],
\edm
and we denote by $ e(b) $ its ground state energy. By the Feynman-Hellman principle, it can be easily seen that $ b \mapsto e(b) $ is a non-decreasing function. In fact, one can prove that it is strictly increasing, whenever there exists a non-trivial minimizer, i.e., for $ b \in (\theo^{-1}, \mu_\beta^{-1}) $. Hence, for any $  \mu_1 < \mu_2 \leq \Theta_0 $,
\[
	\frac{\esec(\mu_1)}{\mu_1} \geq \frac{\esec(\mu_2)}{\mu_2}.
	\]
Multiplying by \(\mu_2>0\) and noticing that \(\frac{\mu_2}{\mu_1} > 1\) and that \( \esec(\mu_1) < 0\), we get
\[ \esec(\mu_1) > \frac{\mu_2}{\mu_1} \esec(\mu_1) \geq \esec(\mu_2).\]
Moreover, for any $ \mu_1, \mu_2 \in (\mu_\beta, \theo) $ and any $ t \in [0,1] $, we have
{\[
	\mathcal{E}_{(1-t)\mu_1 + t \mu_2}[\psi, \mathcal{S}_\beta] = (1-t) \mathcal{E}_{\mu_1}[\psi, \mathcal{S}_\beta] + t \mathcal{E}_{\mu_2}[\psi, \mathcal{S}_\beta] \geq (1-t)  {E}_{\mathrm{sector}}(\mu_1)  + t  {E}_{\mathrm{sector}}(\mu_2),
\]}
and minimizing over \( \psi  \), we get the concavity of
 \(\mu\mapsto \esec(\mu) \).   

Let us now prove \eqref{eq: msec bound 1}, assuming again that \(\mu_\beta<\mu<\Theta_0\). Consider a partition of unity \(\chi_1^2+\chi_2^2=1\) on \(\R^2\) such that \({\rm supp}\,\chi_1\subset \mathcal{B}_{R}(\mathbf{0}) \), $ \chi_2(\mathbf{0}) = 0 $ and \(|\nabla\chi_1|+|\nabla\chi_2|\leq C R^{-1}\), for some given $ R \in \R^+ $ and some finite $ C> 0 $ independent of $ R $.
For any minimizing sequence $ \lf\{ \psi_n \ri\}_{n \in \N} \subset W^{1,2}_{\Fb}(\S) $ of $ \fsec $, we have the decomposition
{\bml{\label{eq: IMS sector}
\mathcal{E}_\mu[\psi_n; \mathcal{S}_\beta] \geq \mathcal{E}_\mu[\chi_1 \psi_n; \mathcal{S}_\beta] + \mathcal{E}_\mu[\chi_2 \psi_n; \mathcal{S}_\beta] - \sum_{i=1}^2\int_{\S_\beta} \diff \xv \: |\nabla\chi_i|^2|\psi_n|^2 	\\
\geq \mathcal{E}_\mu[\chi_1 \psi_n; \mathcal{S}_\beta] + \mathcal{E}_\mu[\chi_2 \psi_n; \mathcal{S}_\beta] -  \tfrac{2c^2}{R^2}\sum_{i=1}^2\int_{\S_\beta} \diff \xv \: |\chi_i\psi_n|^2 
=: \widetilde{\mathcal{E}}_{\mu} [\chi_1 \psi_n; \mathcal{S_\beta}] +  \widetilde{\mathcal{E}}_{\mu}[\chi_2 \psi_n; \mathcal{S}_\beta],
}
where $ \widetilde{\mathcal{E}}_{\mu}[\psi; \mathcal{S}_\beta] : = \mathcal{E}_\mu[\psi; \mathcal{S}_\beta] - \frac{2c^2}{R^2} \lf\| \psi \ri\|_2^2 $. By completing the square, we may rewrite
\[
	\widetilde{\mathcal{E}}_{\mu}[\chi_1\psi; \mathcal{S}_\beta] = \int_{\mathcal{S}_\beta} \diff \xv \: \left[ \lf| \lf(\nabla + \ii \fv \ri) \chi_1\psi \ri|^2 + \tfrac{1}{2} \mu \lf( |\chi_1\psi|^2 - 1 - \tfrac{2c^2}{\mu R^2} \ri)^2 \ri] - \tfrac{1}{2} \mu \lf| \mathrm{supp} \,\chi_1\psi \ri| \lf(  1 + \tfrac{2c^2}{\mu R^2} \ri)^2,
\]
while, by \cref{lem:spectrum} and the fact that $ \mathbf{0} \notin \mathrm{supp}(\chi_2 \psi_n) $,
\bdm
\widetilde{\mathcal{E}}_{\mu}[\chi_2 \psi_n; \mathcal{S}_\beta] \geq \int_{\S_\beta} \diff \xv \lf[ \lf( \Theta_0 - \mu - \tfrac{2 c^2}{R^2} \ri) |\chi_2 \psi_n|^2 + \tfrac{1}{2} \mu |\chi_2 \psi_n|^4 \right].
\edm
Using \eqref{eq: Esector negative} and \eqref{eq: IMS sector}, dropping all positive terms, we get}
\bdm
	{\lf( \Theta_0 - \mu - 2 \tfrac{C^2}{R^2} \ri) \lf\| \chi_2 \psi_n \ri\|^2  \leq \tfrac{1}{2} \pi \mu R^2 \lf(  1 + \tfrac{2c^2}{\mu R^2} \ri)^2.}
\edm
Choosing $ R \sim (\theo - \mu)^{-1/2} $, so that $ \Theta_0-\mu-\frac{c^2}{R^2} = \frac{\Theta_0-\mu}{2} $, we obtain that
\bdm
		{\lf\| \chi_2 \psi_n \ri\|^2_2  \leq \tfrac{C}{\theo - \mu}.}
\edm
Using such a bound to estimate from below $ \widetilde{\mathcal{E}}_{\mathrm{sector}} [\chi_2 \psi_n] $, we deduce that
\bdm
	 \tfrac{1}{2} \mu \lf\| |\chi_1 \psi_n|^2 - 1 - \tfrac{2c^2}{\mu R^2} \ri\|_{L^2(\mathrm{supp}(\chi_1))}^2 \leq \tx\frac{C}{\theo - \mu},
\edm 
which in turn yields 
\bdm
	{\lf\| \chi_1 \psi_n \ri\|^2_2 \leq \tfrac{C}{\theo - \mu}.}
\edm 
Hence,
\bdm
	{\lf\| \psi_n \ri\|^2_2 = \lf\| \chi_1 \psi_n \ri\|^2_2 + \lf\| \chi_2 \psi_n \ri\|^2_2 \leq \tfrac{C}{\theo - \mu}},
\edm
{which can be used to bound from below $ \mathcal{E}_\mu[\psi_n; \mathcal{S}_\beta] $. This, in combination with \eqref{eq: Esector negative}, yields}
\bdm	
	\lf\|  \lf(\nabla + \ii \fv \ri) \psi_n \ri\|_2^2 + \lf\| \psi_n \ri\|^2_2 +  \lf\| \psi_n \ri\|^4_4 \leq \tfrac{C}{\theo - \mu},
\edm
which guarantees the convergence of the minimizing sequence via a standard (compactness) argument for any $ \mu \in (\mu_\beta, \theo) $ and the fact that any limit point $ \msec $ satisfies \eqref{eq: msec bound 1}.
\end{proof}

Another important property of any minimizer of  $ \fsec  $ is its decay in the distance from the corner, which is stated in \cite[Prop. 6.1]{BNF},  and is another form of Agmon decay.

\begin{proposition}[Agmon decay of $ \msec $]
\label{prop:decay}
\mbox{}	\\
Let \eqref{eq: condition} hold and let $ \eps_0 \in (0,1) $.Then, for any minimizer $ \msec $ of  $ \fsec $  and for any $ \mu \in (\mu_{\beta}, \theo) $,
\beq
	\int_{\S_\beta}  \diff \xv \: \lf| \msec(\xv) \ri|^2 \ee^{\varepsilon_0\sqrt{\Theta_0-\mu}\,|\xv|} \leq C_{\mu},
\eeq
where \(C_\mu\) is a constant depending on \(\mu\).
\end{proposition}

\section{Proof of \cref{theo: main}}\label{sec: proof of main thm}

We are now in position to prove our main result, whose argument is divided into two parts: first we prove a lower bound to the energy in the sector in terms of corner effective problem, and then we derive a matching upper bound.  

\subsection{Lower bound}

The main result proven here is the following.

\begin{proposition}[Lower bound]
	\label{pro: lower bound}
	\mbox{}	\\
	Let \eqref{eq: condition} hold.   Then,
	\beq
		\label{eq: lower bound}
		\liminf_{\mu \to \theo^-} \esec(\mu) \geq \edir(\theo).
	\eeq
\end{proposition}

\begin{remark}[Boundedness of $ \esec(\theo) $]
	\label{rem-semi-bd*}
	\mbox{}	\\
By \cref{prop:mu=Th0}, the lower bound in \eqref{eq: lower bound} also implies that \(\liminf_{\mu \to \Theta_0^{-}} \esec(\mu) \) is finite. However, this is not enough to deduce that \( \esec(\theo) \) is finite or that the functional \( \E_{\theo}[\psi; \S_{\beta}] \)  is bounded from below, simply because the continuity of this quantity across the threshold is not obvious. 
\end{remark}

\begin{proof}
Let {$ (L,\ell) \in \mathscr{A} $ and let}   \(0<\eta<\ell/2\) to be chosen later. Consider the domain \(\Gamma_\beta(L-\eta,\ell-\eta)\) and a cut-off function \(\chi_\eta\in C_c^\infty(\R^2;[0,1])\) such that
	\bdm
		\chi_\eta=
		\begin{cases}
		1	&	\mbox{on }\Gamma_\beta(L-2\eta,\ell-2\eta)	\\
		0	&	\mbox{on }\overline{\Gamma_\beta(L,\ell)}\setminus \Gamma_\beta(L-\eta,\ell-\eta)
		\end{cases}
	\edm
	and
	\bdm
		\lf|\nabla\chi_\eta \ri|^2 + \lf| \Delta\chi_\eta \ri| \leq C\eta^{-2}.
	\edm
	Let also $ \mu $ be such that \(0<\Theta_0-\mu \ll 1\). For any minimizing $ \msec $ we recall the bound \eqref{eq: infty} yielding $ |\msec| \leq 1 $. Consider now the trial function
	{\beq
		\psi_{\mathrm{trial}} = \chi_\eta \msec \in \domcord,
	\eeq
	where $\domcord$ is defined in \eqref{eq: def domain Dirichlet corner}.}
Then, we have
\bml{
	\label{eqp: lb 1}
	\ecordl(\theo) \leq \E_{\theo}\big[\psi_{\mathrm{trial}}; \corner \big] \\
	=\int_{\Gamma_\beta(L,\ell)} \diff \xv \: \left\{ \lf|\chi_\eta (\nabla+\ii\Fb) \msec \ri|^2 -\Theta_0 |\chi_\eta \msec|^2+ \tfrac{1}{2} \Theta_0 |\chi_\eta \msec|^4 - \lf(\chi_\eta\Delta\chi_\eta \ri) |\msec|^2  \right\} \\
	\leq \int_{\Gamma_\beta(L,\ell)} \diff \xv \: \left\{ \lf|\chi_\eta (\nabla+\ii\Fb) \msec \ri|^2 -\Theta_0 |\chi_\eta \msec|^2+ \tfrac{1}{2} \Theta_0 |\chi_\eta \msec|^4 + C \eta^{-2} |\msec|^2  \right\}.
	}
We  are going to estimate from above the r.h.s. of the above expression. Remembering that \(0\leq \chi_\eta\leq 1\) and \(\Theta_0-\mu>0\), we can write
\[\begin{aligned}
\lf|\chi_\eta (\nabla+\ii\Fb) \msec \ri|^2	&	\leq \lf|(\nabla+\ii\Fb) \msec \ri|^2,\\
|\chi_\eta \msec|^4&\leq  |\msec|^4,\\
-\Theta_0|\chi_\eta \msec|^2&=-\mu|\msec|^2+\mu(1-\chi_\eta^2)|\msec|^2-(\Theta_0-\mu)|\chi_\eta \msec|^2\\
&\leq -\mu|\msec|^2+\mu(1-\chi_\eta^2)|\msec|^2. \end{aligned}\]
So we get from \eqref{eqp: lb 1}
\begin{multline}
	\label{eqp: lb 2}
	\ecordl(\theo) \leq  \E_{\mu}\big[\msec; \corner \big] +\int_{S_{\beta}} \diff \xv \: \left\{ \tfrac12 (\Theta_0-\mu) |\msec|^4 + \mu(1-\chi_\eta^2)| \msec |^2+C \eta^{-2} |\msec|^2  \right\}	\\
	\leq  \E_{\mu}\big[\msec; \S_{\beta} \big] +\int_{S_{\beta}} \diff \xv \: \left\{ \tfrac12 (\Theta_0-\mu) |\msec|^4 + \mu(1-\chi_\eta^2)| \msec |^2+C \eta^{-2} |\msec|^2  \right\} \\
	+ \OO(\ell^{-\infty}),
\end{multline}
where we used \cref{prop:decay}.
From the variational equation of $ \msec $ and after an integration by parts, we find
\[
	\int_{\S_\beta} \diff \xv \: |\msec|^4 = - 2 \esec(\mu),
\]
so that
\[
	\E_{\mu}\big[\msec; \S_{\beta} \big] + \int_{\S_{\beta}} \diff \xv \: \tfrac12 (\Theta_0-\mu) |\msec|^4 = \lf( 1- (\Theta_0-\mu) \ri) \esec(\mu).
\]
By the exponential decay of \(\ \msec \) in  \cref{prop:decay}, we also get
\[
	\int_{S_{\beta}} \diff \xv \:  \mu(1-\chi_\eta^2)| \msec |^2 =  \OO(\ell^{-\infty}),
\]
and, by \cref{prop:min},
\[ C \eta^{-2}\int_{S_{\beta}} \diff \xv \:  |\msec|^2  \leq \tfrac{C}{\eta^2(\Theta_0-\mu)}.
\]
Collecting the foregoing estimates and plugging them in \eqref{eqp: lb 2}, we obtain
\[
	\ecordl(\theo) \leq \lf( 1- (\Theta_0-\mu) \ri) \esec(\mu) +\tfrac{C}{\eta^2(\Theta_0-\mu)} + \mathcal O(\ell^{-\infty}) .\]
Choosing now \(1\ll \eta\ll\ell\) (recall that \(\ell\ll L\ll\ell^a\) for \(a>1\)) and letting $ \ell \to +\infty $, we get by \cref{lemma:mu=Th0} that
\[
	\edir(\theo) \leq \lf( 1- (\Theta_0-\mu) \ri) \esec(\mu). \]
The l.h.s. in the foregoing inequality is independent of \(\mu\) and therefore we can take the $ \liminf $ for $ \mu \to \theo^- $ to get the result.
\end{proof}

\subsection{Upper bound}

We now address the counterpart of \cref{pro: lower bound} by proving the following result.

\begin{proposition}[Upper bound]
	\label{pro: upper bound}
	\mbox{}	\\
	Let \eqref{eq: condition} hold. Then,
	\beq
		\label{eq: upper bound}
		\limsup_{\mu \to \theo^-} \esec(\mu) \leq \edir(\theo).
	\eeq
\end{proposition}

\begin{remark}\label{rem-semi-bd**}
	\mbox{}	\\
	By a direct inspection of the proof, one can easily realize that the same argument yields the boundedness from above of $ \esec(\theo) $, via the inequality $ \esec(\theo) \leq \ecordl(\theo) $.  As pointed out in \cref{rem-semi-bd*}, the boundedness from below of $ \esec(\theo) $ however does not follow as easily (see the proof of \cref{prop:sb-theta0}). 
\end{remark}

\begin{proof}[Proof of \cref{pro: upper bound}] 
We use a minimizer \( \mdirl \in \domcord \) {defined in \eqref{eq: def domain Dirichlet corner}} of $ \E_{\Theta_0}[\psi,\corner] $ with zero Dirichlet conditions as a trial state for the functional $ \E_\mu[\psi; \S_{\beta}] $, which is possible because any such function can be extended by \(0\) on \(\S_\beta\setminus\corner \). Assuming that {$ (L,\ell) \in \mathscr{A} $,} \(L,\ell>2\) and \(\mu_\beta<\mu<\Theta_0\), we can write
\bml{
	\esec(\mu) \leq \E_{\mu}\big[ \mdirl; \S_{\beta} \big]  = \E_{\mu}\big[ \mdirl; \corner \big] 
	\leq \E_{\theo}\big[ \mdirl; \corner \big]  \\
	+ (\Theta_0-\mu)\int_{\Gamma_\beta(L,\ell)} \diff \xv \: |\mdirl|^2 
	\leq \ecordl(\theo) + (\Theta_0-\mu)\int_{\Gamma_\beta(L,\ell)} \diff \xv \: |\mdirl|^2.
}
Taking \(\mu \to \Theta_0^{-} \) {for fixed $ (L,\ell) $}, we get
\[
	\limsup_{\mu \to \Theta_0^-} \esec(\mu) \leq \ecordl(\theo).
\]
This is true for all \(L,\ell>2\), so taking now the limit as \( \ell\to+\infty\), with $ \ell \ll L \ll \ell^a $, $ a > 1 $,  we get the result via \cref{lemma:mu=Th0}.
\end{proof}

\subsection{Existence of a ground state at the threshold}

{In the next two propositions we show that \(\esec(\theo)\) is finite (\cref{prop:sb-theta0}) and that a ground state of $E_{\mathrm{sector}}(\Theta_0)$ exists (\cref{prop:ex-min}). The combination of the two results provides the proof of \cref{corollary}.}

\begin{proposition}[Sector energy at the threshold]
	\label{prop:sb-theta0}
	\mbox{}	\\
We have \(\esec(\theo)=\ecor^{\rm D}(\theo)\). In particular,  \(\esec(\theo)\) is finite and $\ecor^{\rm D}(\theo)<0$.
\end{proposition}
\begin{proof}
Consider a minimizing sequence \(( \lf\{ u_n \ri\}_{n \in \N}\subset W^{1,2}_{\mathbf F}(\mathcal S_\beta)\) for {$E_{\mathrm{sector}}(\Theta_0)$; that is
\[\lim_{n\to+\infty}\mathcal{E}_{\Theta_0}[u_n, \mathcal{S}_\beta]=\esec(\theo).\]
Fixing \(n\geq 1\) and \(\mu\in(\mu_\beta,\Theta_0)\), we get
\[ \mathcal{E}_{\Theta_0}[u_n, \mathcal{S}_\beta] \geq  \mathcal{E}_{\mu}[u_n, \mathcal{S}_\beta] +(\mu-\Theta_0)\int_{\mathcal S_\beta}|u_n|^2\dd \xv\geq \esec(\mu)+(\mu-\Theta_0)\int_{\mathcal S_\beta}|u_n|^2\dd \xv.\]
Taking \(\mu\to\theo^-\), we get
\[ \mathcal{E}_{\Theta_0}[u_n, \mathcal{S}_\beta] \geq \liminf_{\mu\to\theo^-}\esec(\mu)\geq\ecor^{\rm D}(\theo),\]}
by \cref{pro: lower bound}.
The above inequality is true for all \(n\geq 1\), so, taking \(n\to+\infty\), we get
\[\esec(\theo)\geq \ecor^{\rm D}(\theo).\]

Since the domain of \(\E_\mu[\cdot; \corner]\) is contained in that of {$\mathcal{E}_{\Theta_0}[\cdot; \mathcal{S}_\beta]$} (see \eqref{eq:gse-corner-D}),  we get for all \(L,\ell>2\) (with $(L,\ell)\in\mathcal A$) that
\[\esec(\theo)\leq \ecordl(\theo)\]
and, by taking the limits \(L,\ell\to+\infty\), we get by \cref{lemma:mu=Th0},
\[\esec(\theo)\leq  \ecor^{\rm D}(\theo).\]
This finishes the proof that \(\esec(\theo)=\ecor^{\rm D}(\theo)\).  That $\esec(\theo)$ is finite follows from  \cref{lemma:mu=Th0}, which asserts that $\ecor^{\rm D}(\theo)$ is finite.  Finally, 
by \cref{prop:min},    $\esec(\theo)<0$,  thereby yielding that $\ecor^{\rm D}(\theo)<0$.
\end{proof}

Knowing that $\esec(\theo)$ is finite,  we prove that a ground state of \(\esec(\theo)\) exists.

\begin{proposition}[Existence of a minimizer of $ \fsec $ at the threshold]
	\label{prop:ex-min}
	\mbox{}	\\
There exists \(\psi_*  \in W^{1,2}_{\Fb}(\mathcal{S}_{\beta})\) such that
\beq \E_{\theo}[\psi_*;\S_{\beta}]=\esec(\theo).\eeq
\end{proposition}

We will construct a minimizer \(\psi_*\) as a limit of  a sequence of minimizers of the functional \( \E_{\theo}[\cdot; \corner]\).  The main ingredient in the construction is in the following lemma.
\begin{lemma}\label{lem:ex-min}
There exist positive constants \(C\) and \(L_0\) such that,  if
\(L\geq L_0\) and
 \(\psi_{L,\ell}\in\domcord \) is a minimizer of the functional \( \E_{\theo}[\cdot; \corner]\) with \(\ell=L^{2/3}\), then 
\beq
	\int_{\Gamma_\beta(L,\ell)} \diff \xv\, \lf\{ \lf|(\nabla+i \Fb)\psi_{L,\ell}\ri|^2+|\psi_{L,\ell}|^2+|\psi_{L,\ell}|^4 \ri\} \leq C,
\eeq
 and
\beq
	\label{eq: power law decay}
	 \int_{\Gamma_\beta(L,\ell)}\diff \xv\,  |\xv|^2 |\psi_{L,\ell}|^4\leq C.
\eeq
\end{lemma}

We give the proof of  \cref{prop:ex-min} and postpone the proof of  \cref{lem:ex-min}.\\

\begin{proof}[Proof of  \cref{prop:ex-min}]
\mbox{}	

{\it Step 1.}
For \(\ell=L^{2/3}\) and \(L\gg 1\),  choose a minimizer {\(\mdirl\in \mathcal{D}^{\mathrm{D}}_{L,\ell}\) of the functional \( \E_{\theo}[\cdot; \corner]\) (see Section \ref{sec: dirichlet corner var prob})}.  Then, \(\mdirl\) satisfies
\[|\mdirl|\leq 1\]
and 
\begin{equation}\label{eq:ex-min-PDE}
-(\nabla+\ii\Fb)^2\mdirl=\Theta_0(1-|\mdirl|^2)\mdirl
\end{equation}
with magnetic Neumann boundary condition on \(t(\xv)=0\) and Dirichlet boundary condition 
on \(\{t(\xv)=\ell\}\cup\{s(\xv)=\pm L\}\). For later use, we mention the following identity that results from  \eqref{eq:ex-min-PDE} after taking the inner product with $\mdirl$ and doing an integration by parts,
\bml{\label{eq:ex-min-0}
\ecordl(\theo)=\int_{\Gamma_\beta(L,\ell)} \diff \xv \bigl( |(\nabla+i \Fb)\mdirl|^2-\theo|\mdirl|^2+\tfrac{\theo}{2}|\mdirl|^4\bigr)\\
=-\tfrac{\theo}{2}\int_{\Gamma_\beta(L,\ell)} \diff \xv
\,|\mdirl|^4.
}
 
{\it Step 2.}
Let \(R>0\) and \(K_R=\mathcal B_R\cap \mathcal S_\beta\).  Choosing \(L,\ell\) sufficiently large,    we get that \(K_R\subset K_{2R}\subset \Gamma_\beta(L,\ell)\).  By Lemma~\ref{lem:ex-min} {and the elliptic estimate \cite[Eq. (4.1.2) and Thm.  4.3.1.4]{Gr}
\bdm
	\|u\|_{H^2(\Omega)}\leq C_*\|\Delta u\|_{L^2(\Omega)}  
\edm
for some $ C_* $ depending on the domain $ \Omega $,} we get that,   \(\mdirl\) is bounded in \(H^2(K_R)\). By Cantor's  diagonal  argument and compactness,  we get a function \(\psi_*\in H^2_{\rm loc}(\mathcal S_\beta )\) such that
\begin{equation}\label{eq:comp-1}
\mdirl\to \psi_* \mbox{ strongly in }H^1(K_R)\mbox{ and weakly in }H^2(K_R).
\end{equation}

{\it Step 3.}
Let us verify that \(\psi_*\in W^{1,2}_{\Fb}(\mathcal{S}_{\beta}) \).  For all \(R>0\),  we have
\[\int_{K_R} \diff\xv\,|(\nabla+\ii\Fb) \mdirl|^2\leq C,\]
and by \eqref{eq:comp-1},
\[ \int_{K_R} \diff\xv\,|(\nabla+\ii\Fb)\psi_*|^2\leq C.\]
We take \(R\to+\infty\) and use monotone convergence,  to deduce that
\[ \int_{\mathcal S_\beta} \diff\xv\,|(\nabla+\ii\Fb)\psi_*|^2\leq C.\]
In a similar fashion,  we establish that 
\[ \int_{\mathcal S_\beta} \diff\xv\,|\psi_*|^2\leq C. \]
This proves that \(\psi_*\in W^{1,2}_{\Fb}(\mathcal{S}_{\beta}) \).

{\it Step 4.}
We use {\eqref{eq: power law decay} in  \cref{lem:ex-min}} to write
\[ \int_{\Gamma_\beta(L,\ell)\setminus K_R} \diff\xv\,|\psi_*|^4\leq \frac{C}{R^2}.\]
By \eqref{eq:ex-min-0},
\[ -\tfrac{2}{\theo} \ecordl(\theo)=\int_{\Gamma_\beta(L,\ell)} \diff\xv\,|\mdirl|^4
=\int_{K_R} \diff\xv\,|\mdirl|^4+\int_{\Gamma_\beta(L,\ell)\setminus K_R} \diff\xv\,|\mdirl|^4,\]
and
\[  \left|\int_{ K_R} \diff\xv\,|\mdirl|^4+\tfrac2\theo \ecordl(\theo)\right|\leq \frac{C}{R^2}.\]
Taking \(L\to+\infty\),  we get by Lemma~\ref{lemma:mu=Th0} and  \eqref{eq:comp-1},
\[ \left|\int_{K_R} \diff\xv\,|\psi_*|^4+\tfrac2\theo \edir(\theo)\right|\leq \frac{C}{R^2} . \]
Taking \(R\to+\infty\),  we get by monotone convergence,
\[ \int_{\mathcal S_\beta} \diff\xv\,|\psi_*|^4=-\tfrac2\theo \edir(\theo).\]

{\it Step 5.}
Since \(\psi_*\in W^{1,2}_{\Fb}(\mathcal S_\beta)\),  we have by \eqref{eq:comp-1} and \eqref{eq:ex-min-PDE} that \(\psi_*\) is a critical point of the functional \(\fsec\).  Thus,
\[\fsec(\psi_*)=-\tfrac{\theo}{2}\int_{\mathcal S_\beta}\diff \xv
\,|\psi_*|^4.\]
By {\it Step  4} and \cref{prop:sb-theta0},  we eventually have
\(\fsec(\psi_*)=\esec(\theo)\).
This proves that \(\psi_*\) is a minimizer of \(\fsec\).
\end{proof}

\begin{proof}[Proof of \cref{lem:ex-min}]
\mbox{}	

{\it Step 1.}
Given \(\chi\in C_c^{\infty}(\mathbb R^2)\), we will use extensively  the identity
\begin{equation}\label{eq:ex-min-1}
\int_{\Gamma_\beta(L,\ell)} \diff \xv \, \lf\{ |(\nabla+i \Fb)\chi\mdirl|^2-\Theta_0|\chi\mdirl|^2+\Theta_0\chi^2|\mdirl|^4 \ri\} =\int_{\Gamma_\beta(L,\ell)}\diff\xv\, |\nabla\chi|^2|\mdirl|^2,
\end{equation}
which is a direct consequence of the variational equation \eqref{eq:ex-min-PDE} satisfied by $ \mdirl $ (we take the inner product with $\chi^2\psi_{L,\ell}^D$ and integrate).

{\it Step 2.}
Choose \(\alpha>0\) such that \(1-\theo-\alpha^2>0\) and choose \(\chi\) such that 
\[ \chi(\xv)=\begin{cases}e^{\alpha t(\xv)}&\mbox{ if } t(\xv)\geq 2\\
0&\mbox{ if }t(\xv)\leq 1\end{cases}.\]
Then, by Lemma~\ref{lem:spectrum} and \eqref{eq:ex-min-1}, we have
\[\int_{t(\xv)\geq 2} \diff \xv \lf\{ (1-\Theta_0)|\chi\mdirl|^2+\Theta_0\chi^2|\mdirl|^4 \ri\} \leq \alpha^2\int_{t(\xv)\geq 2} \diff\xv\, |\chi|^2|\mdirl|^2+C. \]
This yields the following estimate,
\begin{equation}\label{eq:ex-min-2}
 \int_{t(\xv)\geq 2}\diff\xv \lf\{ |\mdirl|^2+|\mdirl|^4 \ri\} e^{ 2 \alpha t(\xv)}\leq C.
\end{equation}

{\it Step 3.}
We estimate the integrals of \(|\Psi|^2\) and \(|\Psi|^4\) in the region \(\{t(\xv)\leq 2\}\).  In \eqref{eq:ex-min-1}, we choose \(\chi\) such that 
\[ \chi(\xv)=\begin{cases}
0&\mbox{ if } 0\leq |s(\xv)|\leq 1\\
|s(\xv)|&\mbox{ if } 2\leq |s(\xv)|\leq L-1\\
L-1&\mbox{ if }|s(\xv)|\geq L-1
\end{cases}.\]
Then, by Lemma~\ref{lem:spectrum} and \eqref{eq:ex-min-1}, we have
\begin{equation}\label{eq:ex-min-3}
\Theta_0\int_{|s(\xv)|\geq 1} \diff \xv\, \chi^2|\mdirl|^4\leq C \lf( \int_{1 \leq |s(\xv)|\leq L-1} \diff\xv \, |\mdirl|^2+ 1 \ri).
\end{equation}
We decompose the integral of \(|\mdirl|^2\) as follows
\bmln{ \int_{1\leq |s(\xv)|\leq L-1} \diff\xv\, |\mdirl|^2 = C \int_{\{1 \leq |s(\xv)|\leq L-1\}\cap \{t(\xv)\leq 2\}} \diff\xv \, |\mdirl|^2 \\
+\int_{\{1\leq |s(\xv)|\leq L-1\}\cap \{t(\xv)\geq 2\}} \diff\xv\, |\mdirl|^2.
}
By \eqref{eq:ex-min-2},
\[\int_{\{1 \leq |s(\xv)|\leq L-1\}\cap \{t(\xv)\geq 2\}} \diff\xv \, |\mdirl|^2\leq C,\]
for some $ C$  independent from \(L,\ell\).  Observing that
\[  \int_{\{1\leq |s(\xv)|\leq L-1\}\cap \{t(\xv)\leq 2\}}\diff\xv\,\frac1{\lf|s(\xv)\ri|^2}\leq C, \] 
we get by Cauchy-Schwarz inequality,
\bmln{ 
	\int_{\{1 \leq |s(\xv)|\leq L-1\}\cap \{t(\xv)\leq 2\}} \diff\xv\, |\mdirl|^2 =\int_{\{2\leq |s(\xv)|\leq L-1\}\cap \{t(\xv)\leq 2\}}\diff\xv\,\frac1{|s(\xv)|} |s(\xv)| |\mdirl|^2 + C \\
	\leq \left(\int_{\{2\leq |s(\xv)|\leq L-1\}}\diff\xv\, s^2(\xv)|\mdirl|^4\right)^{1/2} + C.
}
Inserting this into \eqref{eq:ex-min-3}, we get
\[ \bigl( \theo-\tfrac12\bigr)\int_{|s(\xv)|\geq2} \diff \xv\, \chi^2|\mdirl|^4\leq C\]
and
\[\int_{\{2\leq |s(\xv)|\leq L-1\}\cap \{t(\xv)\leq 2\}} \diff\xv \, |\mdirl|^2\leq C.\]
Eventually,  by using  \eqref{eq:ex-min-2} and the uniform bound \(|\mdirl|\leq 1\), we deduce that for \(L,\ell\gg 1\),
\begin{equation}\label{eq:ex-min-4}
 \int_{\Gamma_\beta(L,\ell)}\diff\xv\, \lf\{ |\mdirl|^2+|\mdirl|^4 \ri\} \leq  C,
 \end{equation}
 and
 \begin{equation}\label{eq:ex-min-4*}
 \int_{\Gamma_\beta(L,\ell)}\diff\xv \,  s^2(\xv) |\mdirl|^4\leq C.
 \end{equation}
Finally,  we use \eqref{eq:ex-min-0} with \(\chi=1\) and get
\begin{equation}\label{eq:ex-min-4**}
 \int_{\Gamma_\beta(L,\ell)}\diff\xv\, \lf|(\nabla+\ii\Fb)\mdirl \ri|^2\leq C.
\end{equation}

{\it Step 4.}
Observing that,  in \(\Gamma_\beta(L,\ell)\),
\[ |\xv|\cos\tfrac{\beta}{2}\leq |s(\xv)|\leq |\xv|,\]
we infer from \eqref{eq:ex-min-4*},
\[ \int_{\Gamma_\beta(L,\ell)}\diff\xv\,   |\xv|^2 |\mdirl|^4\leq \frac{ C}{\cos^2\frac{\beta}{2}}. \] 
\end{proof}

\subsection{Proof of \cref{theo: main} and \cref{corollary}}

{By \cref{pro: lower bound} and \cref{pro: upper bound},  we have
\[ \liminf_{\mu\to\theo^+}\esec(\mu)= \limsup_{\mu\to\theo^+}\esec(\mu)=\ecor^{\rm D}(\theo).\]
This proves that $\lim_{\mu\to\theo^+}\esec(\mu)$ exists. Moreover,  by \cref{prop:mu=Th0}, 
$\lim_{\mu\to\theo^-}\ecor(\mu)$ exists too and equals the r.h.s. of the above identity.   \cref{theo: main},  thanks to \cref{prop:sb-theta0}.  Finally, the statement of \cref{corollary} follows from \cref{prop:sb-theta0} and  \cref{prop:ex-min}.}

\bigskip

\begin{footnotesize}
\noindent
\textbf{Acknowledgments.} The authors acknowledge the support of the Istituto Nazionale di Alta Matematica “F. Severi”, through the Intensive Period “INdAM Quantum Meetings (IQM22)” at Politecnico di Milano. A.K. is partially supported by  CUHKSZ, grant no.  UDF01003322. M.C. acknowledges the supports of PNRR Italia Domani and Next Generation Eu through the ICSC National Research Centre for High Performance Computing, Big Data and Quantum Computing, and of the MUR grant ``Dipartimento di Eccellenza 2023-2027'' of Dipartimento di Matematica, Politecnico di Milano.
\end{footnotesize}


\begin{thebibliography}{aaaaa9999}

\bibitem[AB-N]{ABN} \textsc{F. Alouges, V. Bonnaillie-N\"{o}el}, Numerical computations of fundamental eigenstates for the Schr\"{o}dinger operator under constant magnetic field, {\it Numer. Methods Partial Differential Eq.} {\bf 22} (2006), 1090--1105.
%
\bibitem[As1]{A20}		\textsc{W. Assaad}, The breakdown of superconductivity in the presence of magnetic steps,   {\it Commun. Contemp. Math. } {\bf 23} (2021), 2050005.
%
\bibitem[As2]{A21}  \textsc{W. Assaad},  
Magnetic steps on the threshold of the normal state,  
{\it  J.  Math.  Phys.} {\bf 61} (2020), 101508.


\bibitem[AG]{AG} \textsc{W.  Assaad,   E. L.  Giacomelli},  A 3D-Schr$\ddot{\rm o}$dinger operator under magnetic steps with semiclassical applications,  {\it Discrete Contin. Dyn. Syst. }{\bf 43}  (2023),   619--660.
%
\bibitem[AKP-S]{AKP} \textsc{W. Assaad, A. Kachmar, M. Persson-Sundqvist},  The Distribution of Superconductivity Near a Magnetic Barrier,   {\it Commin. Math. Phys.} {\bf 366} (2019), 269--332.
%
\bibitem[Bon]{Bo}	\textsc{V. Bonnaillie}, On the fundamental state energy for a Schr\"{o}dinger operator with magnetic field in domains with corners, {\it Asymptot. Anal.} {\bf 41} (2005), 215--258.
%
\bibitem[B-NF]{BNF} \textsc{V. Bonnaillie-No\"el, S. Fournais}, Superconductivity in Domains with Corners, {\it Rev. Math. Phys.} \textbf{19} (2007), 607--637.
%
\bibitem[B-NFKR]{BNFKR} 	\textsc{V. Bonnaillie-No\"el, S. Fournais, A. Kachmar, N. Raymond}, Discrete spectrum of the magnetic Laplacian on perturbed half-planes, preprint {\it arXiv:2208.13646 [math.SP]} (2022).
%
\bibitem[Cor]{Co} \textsc{M. Correggi},  Surface effects in superconductors with corners, 
{\it Bull. Unione Mat.  Ital. } {\bf 14} (2021), 51--67. 

\bibitem[CDR]{CDR} \textsc{M. Correggi, B. Devanarayanan, N. Rougerie}, Universal and shape dependent features of surface superconductivity, {\it Eur. Phys. J. B} {\bf 90} (2017), 231. 

\bibitem[CG1]{CG1} \textsc{M. Correggi, E.L. Giacomelli}, \textit{ Surface superconductivity in presence of corners}, Rev. Math. Phys. {\bf 29}, 1750005 (2017).

\bibitem[CG2]{CG2} \textsc{M. Correggi, E.L. Giacomelli}, Effects of Corners in Surface Superconductivity, {\it Calc. Var. Partial Differential Equations} {\bf 60} (2021), 236.
%
\bibitem[CG3]{CG3}  \textsc{M. Correggi, E.L. Giacomelli}, Almost Flat Angles in Surface Superconductivity, {\it Nonlinearity} {\bf 34} (2021), 7633--7661.
%
\bibitem[CR1]{CR1} \textsc{M. Correggi, N. Rougerie}, On the Ginzburg-Landau Functional in the Surface Superconductivity Regime, \textit{Commun. Math. Phys.} \textbf{332} (2014), 1297--1343; erratum {\it Commun. Math. Phys.} {\bf 338} (2015), 1451--1452.
%
\bibitem[CR2]{CR2} \textsc{M. Correggi, N. Rougerie}, Boundary Behavior of the Ginzburg-Landau Order Parameter in the Surface Superconductivity Regime, {\it Arch. Rational Mech. Anal.} {\bf 219} (2015), 553--606.
%
\bibitem[CR3]{CR3}	\textsc{M. Correggi, N. Rougerie}, Effects of boundary curvature on surface superconductivity, {\it Lett. Math. Phys.}  {\bf 106} (2016), 445--467.
%
\bibitem[ELP-O]{ELP} \textsc{P. Exner, V. Lotoreichik, A. P\'{e}rez-Obiol}, On the Bound States of Magnetic Laplacians on Wedges, {\it Rep. Math. Phys.} {\bf 82} (2018), 161--185.
%
\bibitem[FH2]{FH1} \textsc{S. Fournais, B. Helffer}, \textit{Spectral Methods in Surface Superconductivity}, Progress in Nonlinear Differential Equations and their Applications \textbf{77}, Birkh\"auser, Basel, 2010.

\bibitem[Gia]{Gia} \textsc{E.L. Giacomelli}, \textit{On the magnetic Laplacian with a piecewise constant magnetic field in $\mathbb{R}^3_+$}, in {\it Quantum Mathematics I}, M. Correggi, M. Falconi  (eds), Springer INdAM Series {\bf 57}, Springer, Singapore, 2023.
%
\bibitem[Gri]{Gr} \textsc{P. Grisvard}, {\it Elliptic Problems in Nonsmooth Domains},  Classics in Applied Mathematics {\bf 69}, SIAM,  2011.
%
\bibitem[HK]{HK} \textsc{B. Helffer,   A. Kachmar}, The density of superconductivity in domains with corners,   {\it Lett. Math. Phys.} {\bf 108} (2018), 2169--2187.

\bibitem[HM]{HM} \textsc{B. Helffer, A. Morame}, Magnetic Bottles in Connection with Superconductivity, {\it J. Funct. Anal.} {\bf 185} (2018), 604--680.

\bibitem[Jad]{Ja} \textsc{H.T. Jadallah}, The onset of superconductivity in domains with corner, {\it J. Math. Phys.} {\bf 42} (2001), 4101.
%
\bibitem[Mir]{M}   \textsc{G. Miranda},  Discrete spectrum of the magnetic Laplacian on almost flat magnetic barriers,  arXiv:2308.14680 (2023).
%
\bibitem[Ray]{Ra} \textsc{N. Raymond}, {\it Bound States of the Magnetic Schr\"{o}dinger Operator}, EMS Tracts in Mathematics {\bf 27}, EMS, 2017.
%
\end{thebibliography}
\end{document}